%% file: arxiv_v2.tex
\documentclass[11pt,letter]{article}

\providecommand{\keywords}[1]{
  \small
  \textbf{\textit{Keywords---}} #1
}
\usepackage{tikz}
\usepackage{comment}
\usepackage{amsfonts,amssymb}
\usepackage{graphicx}
\usepackage{textcomp}
\usepackage{mathtools}
\usepackage{subcaption}
\usepackage{multirow}
\usepackage{array}

\usepackage{amsthm}

\newtheorem{thm}{Theorem}
\newtheorem{theorem}{Theorem}

\newtheorem{lemma}[thm]{Lemma}

\newtheorem{claim}[thm]{Claim}

\theoremstyle{plain}
\newtheorem*{lemma*}{Lemma} 

\usepackage{todonotes} 

\DeclareMathOperator*{\argmax}{arg\,max}
\DeclareMathOperator*{\argmin}{arg\,min}

\newcommand{\cA}{\mathcal{A}}

\newcommand{\hY}{\widehat{Y}}

\newcommand{\eps}{\epsilon}
\newcommand{\hH}{\widehat{H}}
\newcommand{\hG}{\widehat{G}}

\usepackage{algorithm}
\usepackage{algpseudocodex}

\usepackage{microtype,booktabs,array}

\newcommand{\opt}{\mathit{opt}}
\newcommand{\OPT}{\mathit{OPT}}
\newcommand{\FFD}{\operatorname{FFD}}
\newcommand{\wt}{\operatorname{wt}}

\begin{document}
\title{Busy Time Minimization with Preemption, Migration, and One Resource Requirement} 

\author{Gruia C\u alinescu\thanks{
The Research Institute of the University of Bucharest, Romania. gcali9999@gmail.com https://orcid.org/0000-0002-0925-9524  Research done in part at Northwestern University and in part at Illinois Institute of Technology
}
\and
 Mozhengfu Liu\thanks{
 Northwestern University, Evanston, IL, USA.  mozhengfuliu2027@u.northwestern.edu  https://orcid.org/0000-0001-6181-9958  Supported by an Adobe award and NSF 2216970 (IDEAL).
 }
 }

\maketitle

\begin{abstract}
We study the Busy Machine Time with Preemption and Migration and One Resource Requirement problem, motivated by energy minimization in cloud data centers. Given unlimited identical-capacity machines and jobs with arrival times, deadlines, processing times, and resource requirements, we allow free preemption and migration at integer times and seek to minimize total machine busy time. The problem is NP-hard, and previous results consist of a $2$-approximation, $2$-competitive algorithm for the case of uniform heights.

We obtain a $5671/2340 \approx 2.424$-approximation algorithm and a $2.5$-competitive online algorithm, both running in $O(n^2 \log n)$ time.
Our methods are based on new absolute performance bounds for the First Fit Decreasing algorithm for Bin Packing, and a new generalization of Span Minimization, the Huge-Tiny Busy Time problem for which we present an exact offline algorithm and an optimal $3/2$-competitive deterministic online algorithm.
\end{abstract}

\section{Introduction}
\label{s_intro}

We study the Busy Time Scheduling problem in the following setting.
We have access to an unlimited number of machines, all having the same capacity.
A set of jobs is given as input, each job characterized by
arrival time, processing time, and deadline, and one resource requirement that we call \emph{height}.
All jobs must be scheduled, and we allow free preemption and migration of jobs from
one machine to another at discrete time points. The sum of the heights of the jobs assigned to a machine at a given time
cannot exceed the capacity of the machine. With these constraints, we aim to minimize the 
so-called \emph{busy machine time}, which is the total time machines are used to accommodate the input jobs.
Note that if the schedule uses several machines at the same time, the usage of each of these machines is added to the busy time.
See Section \ref{s_prel} for formal definitions. We call our problem BMTPMORR, for
Busy Machine Time with Preemption and Migration and One Resource Requirement.

The BMTPMORR problem is NP-hard by a trivial reduction from Bin Packing.
Even with uniform heights, BMTPMORR is NP-hard as the construction given by \cite{CaoFLMRU22,cao_et_al:LIPIcs.ISAAC.2022.36} for 
Active Time also works for Busy Time (this can be checked with some effort).
With unlimited machine capacity, at any given time, at most one machine is used and the goal
becomes minimizing the so-called \emph{span}. Span minimization can be done in polynomial time 
by the ``greedy'' algorithm from Theorem 6 of
\cite{CKM17}. Based on this algorithm, \cite{CKM17} (Theorem 7) 
obtains a rather simple $2$-approximation algorithm for the variant
of BMTPMORR with uniform heights, and this is the best known so far.

We are also interested in the online setting (see \cite{BEY98} for an introduction to online algorithms),
where a job $j$ becomes known to the algorithm at its arrival time $a(j)$.  It is an easy observation that
the Span Minimization algorithm of \cite{CKM17} can be easily adapted to work in the online setting;
this will also follow from our algorithm in Section \ref{s_hp}. In other words, Span Minimization with preemption
has an exact online algorithm. The $2$-approximation algorithm for uniform heights of \cite{CKM17}  
also works in the online setting and gives a $2$-competitive algorithm for BMTPMORR with uniform heights.

Approximation and competitive ratio of 3 for BMTPMORR are straightforward:
one uses an online version of the algorithm from Theorem 6
of \cite{CKM17} for Span Minimization,
and then uses AnyFit \cite{coffman1997approximation}
as a bin-packing procedure for the jobs assigned to every timeslot
(as described here, this can be a pseudopolynomial-time algorithm,
but with a bit of effort similar to Appendix \ref{s_pt}
it can be transformed into a polynomial-time algorithm).
We give the proof of this ratio of 3 in Appendix \ref{s_3}.
We keep this approach of separating temporal scheduling from
machine packing, but use a better algorithm for packing
and a  new, custom algorithm for temporal scheduling.

For Bin Packing, one can do much better than AnyFit; however there is the challenge 
that our algorithms  may not assign to one
timeslot the same jobs as an optimum schedule assigns,
so comparing directly to the optimum of one timeslot did not help our analysis.
We found First Fit Decreasing (FFD)\footnote{FFD sorts the items in non-increasing order of size and then places each item,
in that order, into the first (lowest-indexed) bin that has enough remaining capacity to hold it,
opening a new bin only when none of the existing bins fit.
}
to be the best suited Bin Packing algorithm for our proof method;
in order to use FFD, we prove new upper bounds for it in Section \ref{s_ffd}.

In order to improve these bounds, we introduce a new problem, HTBTPM (Huge-Tiny Busy Time with Preemption and Migration),
that generalizes the Span Minimization problem. HTBTPM has two types of jobs: those with huge resource requirement,
so no two huge jobs can be scheduled
simultaneously on the same machine,
and those with tiny resource requirements,
which can all fit on any machine, even together with a huge job. HTBTPM may be of independent interest. 
There is a common method of separating jobs with heights greater than $\epsilon$ (called large jobs) from jobs with heights smaller than or equal to $\epsilon$ (called small jobs) and hereafter the two types of jobs are packed differently. See \cite[Chapter~9]{Vazirani2001_bin_packing} for an example. 
Note that HTBTPM is not a method, but a problem.
We are able to solve exactly HTBTPM using a greedy algorithm
that generalizes the algorithm of \cite{CKM17} for Span Minimization.
While the algorithm itself is rather simple,
its proof of correctness turns out to require involved techniques based on a matching lower bound,
as seen in Section \ref{s_hp}.
Minor changes in the algorithm allow us to obtain a 
deterministic $1.5$-competitive online algorithm for HTBTPM; 
this is best possible (Theorem  \ref{t_32} in Section \ref{s_prel}).  The proof of the competitive ratio uses the same lower bound
and much of the machinery used in analyzing the exact greedy algorithm.

The HTBTPM algorithms 
determine the temporal assignment, while FFD constructs a feasible packing under the original machine-capacity constraints. 
Combined with our bounds on FFD, the HTBTPM upper bounds give us a $5671/2340 < 2.424$ ratio for the approximation algorithm,
and a $2.5$ competitive ratio for the online algorithm. Both algorithms can be made to run in time $O(n^2 \log n)$
(Appendix \ref{s_pt}).
Intriguingly, our approximation ratio is not much better than our competitive ratio; but for uniform heights too,
the best known ratio of 2 holds for both settings \cite{CKM17}.
We also present in Appendix \ref{s_lb}
a lower bound of $1.5$ for online deterministic
algorithms  in the uniform heights case. 
Here we assume that the machine capacity is 1 and every
job has height $1/g$, and the above bound holds for
$g \geq 4$.

\subsection{Related Work}
With preemption and migration, 
\cite{sarpatwar2023preemptive} study minimizing the number of identical machines that can accommodate
a given set of jobs,
in which jobs require $d>1$ resources  and have $d$ requirements. 
The machines have capacity constraints for each of these resources.
Like in this work, the jobs are flexible and have variable
processing times. Note that their objective is not  busy time, since they do not
add over time the number of machines used in each timeslot,
but instead take the maximum, over time, of
the number of machines used in each timeslot.
They obtain a $O(\log{d}\log^*{T})$-approximation, where $T$ is the maximum deadline of any job.

Driven by demands of cloud computing where physical machines (corresponding to machines) in large data centers are rented for hosting virtual machines (corresponding to jobs), minimizing the total energy consumption and maximizing efficiency of the physical machines \cite{MASDARI2016106} received much attention.
In the online setting of rigid job scheduling, i.e. each rigid job is processed fully from its arrival to deadline, Mellou \emph{et al.} \cite{Mellou2024migration} investigate how migration (preemption is not needed here as jobs are rigid) helps break the barrier of the lower bound for any
non-migratory scheduling algorithm.
More related works on Busy Machine Time can be found in the survey \cite{CL19}.

If preemption is not allowed, Busy Time Minimization with One Resource Requirement was studied by Khandekar \emph{et al.} \cite{KSST15} 
and they obtain a $5$-approximation algorithm. When the jobs have the same height, Chang et al. \cite{CKM17} show that
the algorithm of \cite{KSST15} is a $4$-approximation, and then gives a $3$-approximation algorithm.

For migration, Kalyanasundaram and Pruhs \cite{KALYANASUNDARAM20012} study the difference between the scheduling with migration and the scheduling without migration on the number of processors used for completing a given set of jobs. 
For online preemptive scheduling, Chen \emph{et al.} \cite{ChenEMSS20} study the throughput maximization problem where a commitment must be made before each job's deadline and each job has an $\epsilon$ slackness. 

For classic online Bin Packing problem, \cite{Angelopoulos2018} studies how advice about the input helps to improve the competitiveness. 
\cite{BDESV18} studies the colored bin packing where items  of
the same color cannot be placed next to each other in the same bin. 
In Demand Strip Packing problem where items have different heights and widths and a strip of a fixed width is given, all the items are packed within the given strip and the objective is to minimize the peak height. Jansen \emph{et al.} \cite{pseudo_DSP} closed the problem in pseudo-polynomial time. Eberle \emph{et al.} \cite{polynomial_DSP} 
closed the problem in polynomial time.

\subsection{Paper Organization}
The next section  provides the preliminaries we use.
Section \ref{s_ffd} gives our first analysis of the 
FFD algorithm. Section \ref{s_hp}
presents our exact and online algorithms for the Huge-Tiny case.
The exact algorithm is used in Section \ref{s_offline&online}
to obtain an algorithm with approximation  ratio $22/9 \approx 2.445$ for the 
Busy Time Minimization with Preemption, Migration, and One Resource Requirement problem.
A very similar analysis using the  online algorithm 
for the Huge-Tiny case gives a
competitive ratio of $2.5$ for Busy Time Minimization with Preemption, Migration, and One Resource Requirement.

Section \ref{s_242} gives a family of examples where
both our algorithms have ratio at least $2.423$.
This is very close to the upper bound of Appendix \ref{s_chatgpt}
for the approximation ratio of our offline algorithm,
which is obtained
by a much more involved analysis of the FFD algorithm.
Appendix \ref{s_pt}  obtains the $O(n^2 \log n)$  running time
for both of our algorithms.

\section{Preliminaries and Lower Bounds} 
\label{s_prel}

Job $j$ has an interval 
$W_j$ (called the {\em window} of $j$) where it can be scheduled. 
The left point of $W_j$ is the arrival time $a(j) $ of job $j$, and the right endpoint is the deadline $d(j)$ of job $j$.
Job $j$ has processing time $p(j) \leq d(j) - a(j)$ and ``height'' $\sigma(j)$. The numbers $a(j),d(j),p(j)$ are all integers.

Without loss of generality, we assume that all
machines have capacity $1$. This work assumes that a machine can start and end running at integer units of time.
Since this work assumes that preemption and migration of jobs is possible and free,  by ``splitting'' a machine's run at
an integer unit of time, we can assume without loss of generality that a machine executes
the same set of jobs for all the duration 
the machine is active. Note that the capacity constraint dictates
that the sum of the heights of this set of jobs does not exceed  $1$.

Job $j$ must be assigned to one or more machines so that the total running time of these machines is $p(j)$,
all these machines are active only inside $W_j$,
and no two of these machines are active at the same time.
The total time (sum of all durations) that the machines run is the {\bf busy time} that this work aims to minimize.

A timeslot is an interval $[i,i+1]$ where $i$ is an integer.
Given an algorithm, we say that a job \emph{opens} a timeslot if it is the first job scheduled  by the algorithm in that timeslot.

A job $j$ is a \emph{unit} job if $p(j) = 1$.
Scheduling a set of unit jobs $J$ gives a function $f:J \rightarrow \{ [i,i+1] \; | \; i \in [0,T-1] \}$
(as above $T = \max_j d(j)$), where $f(j) \subseteq W_j$. 
We say that a schedule of set of unit jobs $J$ \emph{covers} a set of timeslots $S$ if
every timeslot in $S$ has a job of $J$ scheduled in it, or in  more mathematical terms,
$S \subseteq  \{ [i,i+1] \; | \; \exists j \in J \; f(j) = [i,i+1] \}$,
where $f$ is the function associated with the schedule.
We say that a set of unit jobs $J$ \emph{can cover} a set of timeslots $S$ if there exists a schedule of $J$
that covers $S$.

Let $v = \sum_j \sigma(j) p(j)$ be the ``volume''. 
Then one lower bound is:
\begin{equation}
    \opt \geq \lceil v \rceil
    \label{lb_volume}
\end{equation} 
Indeed, we can assume (by splitting) that any machine is used for exactly one timeslot.
Take any machine (busy for one timeslot) $M$.
The set of jobs processed by the machine $M$ has the total height at most $1$
and each job in it is processed for one timeslot by the machine $M$.
Therefore, for any feasible (optimal) scheduling, one unit of busy time corresponds to a
volume $\sum_{j: j\text{ processed in } M} \sigma(j) \leq 1$.
Summing up for all machines, the total busy time due to the scheduling 
is the total number of machines and is greater than or equal to
$\sum_M \sum_{j: j\text{ processed in } M} \sigma(j) = \sum_j \sum_{M: M \text{ processes }j} \sigma(j) = \sum_j \sigma(j)p(j)$. 
Due to the integrality, lower bound (\ref{lb_volume}) follows. 

Consider $Q$,  a set of timeslots of minimum size, so that we can
schedule all the jobs in $Q$ (assuming unlimited machine capacity). 
Note the lower bound:
\begin{equation} 
\opt \geq |Q|
\label{lb_span}
\end{equation}
Indeed,
for any feasible scheduling, the input jobs are processed during at least $|Q|$
timeslots by the definition of $Q$. 
Since for each of the timeslots some job is processed by some machine,
there are at least $|Q|$ units of busy time in the scheduling. Lower bound (\ref{lb_span}) follows. 

Let $J_H$  be the set of {\em huge} jobs of height
strictly bigger than  $1/2$. 
Let $J_B$  be the set of {\em big} 
jobs of height strictly bigger than $1/3$ (so $J_H \subseteq J_B$). 

Now we define our third lower bound:
\begin{equation}
     h := \sum_{j \in J_H} p(j)   \leq \opt,
    \label{lb_tall_jobs}
\end{equation}
as each of the jobs $j$ requires $p(j)$
distinct machines (one machine cannot accommodate two of these jobs).

We will make use of the following lower bound. Let $b:= \sum_{j \in  J_B \setminus J_H} p(j)$ 
\begin{equation}
\label{e_JB}
    \opt \geq  \frac{h+b}{2}
\end{equation}
as each machine processes at most two distinct jobs from $J_B$ in one timeslot. 

Given an BMTPMORR instance, we can construct a  HTBTPM instance 
where all the non-huge jobs are treated as tiny.  The optimum of the new instance
cannot exceed the optimum of the original instance, and this gives us another lower bound on $\opt$.

Suppose a job $j$ has already been assigned to $l(j) < p(j)$ timeslots. We call $d(j) - (p(j) - l(j))$ 
the \emph{starting deadline} of $j$, as $j$ must be assigned to all the timeslots between 
 $d(j) - (p(j) - l(j))$  and $d(j)$ if no further work is performed before that time. 

\begin{theorem}
    \label{t_32}
    The competitiveness of any (deterministic) online algorithm for the Huge-Tiny setting is at least $3/2$. 
\end{theorem}
\begin{proof}
The adversary releases a huge job at time $0$, deadline $2$, and processing time $1$.
The adversary releases a tiny job at time $0$, deadline $3$, and processing time $2$.
If the algorithm chooses to open a bin at timeslot $[0,1]$, then the adversary releases
a tiny rigid job at time $1$ with processing time $2$.
The algorithm incurs a cost of at least $3$,
while a solution of cost $2$ exists - open bins at timeslots $[1,2]$ and $[2,3]$.
If the algorithm does not open the bin at timeslot $[0,1]$, then the adversary releases
a huge rigid job with processing length $1$ at time $1$. The algorithm incurs a cost of
at least $3$, while
a solution of cost $2$ exists - open one bin at each of the timeslots $[0,1]$ and $[1,2]$.
\end{proof}

Note that our work from Section \ref{s_hp} implies that this lower bound is tight.

\section{Some FFD properties}
\label{s_ffd}

We consider Bin Packing, and give bounds on the performance of FFD 
(analyzed with respect to optimum in \cite{johnson1973near,BAKER198549,DLHT13})
in terms of $v',b', h'$ where $v'$ is the volume, 
$h'$ is the number of huge jobs, and $b'$ is the number of big-but-not-huge jobs.
FFD was shown to have approximation ratio (absolute, not asymptotic) of $3/2$ in \cite{simchilevi1994new} 
and the proof of this ratio given in Theorem 18.7 of  \cite{KV12} gave us some useful ideas that we use below.

\begin{lemma}
\label{lem: ffd}
    Consider a Bin Packing instance where $h'$ denotes the number of huge items, $b'$ is the number of items that are big but not huge,
    and $v'$ is the total volume.
    Let $a'$ be the output of FFD.
    If $h' = 0$, then $a' \leq 1 + \frac{1}{18} b' + \frac{4}{3} v'$.  Also if $h' =0$, then $a' \leq 1 + \frac32 v'$.
    If $h' > 0$, then $a' \leq h' +  \frac{1}{18} b' + \frac{4}{3} v'$. Also, if $h' >0$, then $a' \leq \frac12 h' + \frac32 v'$.
\end{lemma}

\begin{proof}
    If the last item is not put in the last bin, remove it from the instance; this does not change $a'$ and cannot increase $h'$, $b'$, or $v'$, so making the proof for the smaller instance is enough.
    On top of it, if the last item (must be in the last bin) is not the first item packed into the last bin, remove it from the instance by the same argument.
    After all the removals, eventually, the last bin contains only one item which is also the last item in the instance.
    Moreover, no huge item is ever removed as it always is the
    first item scheduled in a bin.

    Let us first consider the case $h' =0$.
    As FFD starts filling up the bins (also see Figure \ref{fig:ffd_h0}), it will have $\beta = \lfloor b'/2 \rfloor$ bins with  two big jobs.
    Then, maybe, one more big job in another bin. The algorithm finishes after using another $\gamma + 1$ bins,
    so $a'  = \beta + \gamma + 1$, where
    $\gamma \geq 0$ since the last bin only has one item.

    Let the last item (which is in the last bin) have size $l$.
    The first $\beta + \gamma$ bins all have volume
    more than $1 - l$, since this last item did not fit there.
    In a first case, $l \leq \frac14$.
    Then $v' \geq \frac34 (\beta + \gamma)$ and the lemma follows.

    In the other case, $\frac14 \leq l$.
    First let us consider the subcase $\gamma$ is  $0$.
    As each of the first $\beta$ bins has (at least) two big jobs,  $v' \geq \frac23 \beta$,
    and also using $b' \geq 2 \beta$, the lemma follows from algebraic manipulations.
    In a second and last subcase, $\gamma > 0$.  Note that $l \leq \frac13$
    as at most one big job is not in the bins $1, 2, \ldots, \beta$.
    Bins $\beta +2, \beta + 3, \ldots, \beta + \gamma$ each have
    at least three jobs,
    and so  each has volume at least $3l$.
    Bin $\beta + 1$ has volume more than $1-l$. The total volume in the first $\beta$ bins
    is at least $\frac23 \beta$,
    and the volume of the remaining $\gamma + 1$ bins is at least $1 + (\gamma -1) 3 l$.
    Thus $v' \geq \frac23 \beta + 3 l \gamma$.
    Then $\frac43 v' + \frac{1}{18}b' \geq \frac89 \beta + 4l \gamma + \frac19 \beta \geq \beta + \gamma$
    and the first inequality of the lemma follows.

    Again, we consider $h' = 0$. Notice that when $a'=1$  the inequality $a' \leq 1 + \frac32 v'$ holds. Assume now  $a' > 1$.
    If the last bin has a job with size smaller than $\frac13$, then each of the first $a'-1$ bins has volume at least $\frac23$,
    or else the last item would fit there. If the last bin has a job with size at least $\frac13$, then all jobs have size
    at least $\frac13$, and since there are no huge jobs, each of the first $a'-1$ bins has
    at least two jobs of size at least $\frac13$, and thus volume at least $\frac23$.
    So in all cases $v' \geq \frac23 (a'-1)$  and thus $1 + \frac32 v' \geq a'$.

    Now we consider the case $h' > 0$. FFD works as follows (also see Figure \ref{fig:ffd_hplus}): First $h'$ bins get one huge job each.
    The $b'$ big-but-not-huge jobs come next. Some may fit in the first $h'$ bins. Some may not, and will go two by two in
    another $\beta$ bins.
    At most one big-but-not-huge job will go in bins further than $h' + \beta$.
    Small jobs follow next, some may fit in the first
    $h' + \beta$ bins, and the others will go in newer bins, resulting in $a' = h' + \beta + \gamma + 1$ bins, where $\gamma = -1$ is possible.
    Let $\hat{b}$ be the big jobs in bins $h' + 1, \ldots, h' + \beta + 1$ (these $\hat{b}$ jobs are big-and-not-huge).
    Let the last item (which is in the last bin) have size $l$.
    The first $h' + \beta + \gamma$ bins all have volume
    more than $1 - l$, since this last item did not fit there.
    In a first case, $l \leq \frac14$. We have $\gamma\geq 0$.
    Then $v' \geq \frac34 (h' + \beta + \gamma)$ and
    therefore $h' + \frac43 v' \geq h' + h' + \beta + \gamma \geq a'$ as $h' \geq 1$.

    In the other case, $\frac14 \leq l$.  First let us consider the subcase $\gamma$ is $-1$ or $0$.
    As each of the  bins $h'+1, \ldots, h' + \beta$ each has (at least) two big-but-not-huge jobs,  each has volume at least $\frac23$.
    The first $h'$ bins each has volume at least $\frac12$.  If $\gamma = -1$,
    then $\beta = 0$ (only one item on the last bin) and the lemma follows.
    If $\gamma=0$,
    the combined volume of the first and last bin is at least 1.
    We conclude $v' \geq 1 + \frac12 (h'-1) + \frac23 (\beta + \gamma)$.
    Using $b' \geq \hat{b} \geq  2 \beta$, the lemma follows from 
    $h' + \frac43 v' + \frac{1}{18}b' \geq 
    h' + \frac23h' + \frac23 + \frac89 (\beta + \gamma) + \frac19 \beta 
    \geq h' + \beta + \gamma + \frac23h' + \frac23 - \frac19 \gamma 
    \geq h' + \beta + \gamma + 1$.

    In a second and last subcase, $\gamma > 0$.  Note that $l \leq \frac13$
    as at most one big-but-not-huge job is in the bins $h' + \beta +1, \ldots, h' + \beta + \gamma + 1$.
    Bins $h' + \beta +2, h' + \beta + 3, \ldots, h' + \beta + \gamma$ each have at least three jobs,
    and so  each has volume at least $\frac34$. The first $h' + \beta$ bins each have volume at  least $\frac23$
    since the last job did not fit there. Bin $h' + \beta + 1$ and the last bin combined have volume at least $1$.
    Thus $v' \geq \frac23 (h' + \beta) + \frac34 (\gamma - 1) + 1$. Again, using $b' \geq 2 \beta$  and $h' \geq 1$,
    the inequality $a'  = h' + \beta + \gamma + 1 \leq h' + \frac{1}{18} b' + \frac43 v'$ follows from algebraic manipulation.

    For the last inequality of the lemma, we look at the same scenario.
    Clearly, the case $a' = 1$ holds, so it remains to consider the case $a'\geq 2$.
    We have $v' \geq (h' + \beta + \gamma) (1-l) + l$.
    If $l \leq \frac13$, we get $v' \geq 1 + \frac23 (h' + \beta + \gamma -1)$
    and therefore
    $\frac12 h' + \frac32 v' \geq \frac12 h' + \frac32 + (h' + \beta + \gamma - 1) \geq \frac12 + \frac32 + (h' + \beta + \gamma - 1)
    = (h' + \beta + \gamma  + 1) = a'$.
    If $l > \frac13$ (the last job is big), then all the jobs are big and $\gamma$ is $-1$ or $0$.
    Moreover, each of the bins $h'+1, h'+2, \ldots, h' + \beta$ has two big jobs and volume
    at least $\frac23$ each. Out of the first $h'$ bins, all have volume more than $1-l$ and at least $\frac12$.
    Then, we have a total volume of at least $\frac23 \beta + \frac12 (h'+\gamma-1) + (1-l) + l = \frac12 + \frac12 h' + \frac12 \gamma + \frac23 \beta$ where $\gamma = -1$ only when $h' = a'\geq 2$.
    Therefore
    $\frac12 h' + \frac32 v' \geq \frac12 h' + \frac34 + \frac34 h' + \frac34 \gamma + \beta = h' + \beta + \frac34 \gamma + \frac14 h' + \frac34 \geq h' + \beta + \gamma + 1$.
    In both cases, the inequality follows.
\end{proof}

\begin{figure}
\centering
\begin{subfigure}{.5\textwidth}
  \centering
  \includegraphics[width=.7\linewidth]{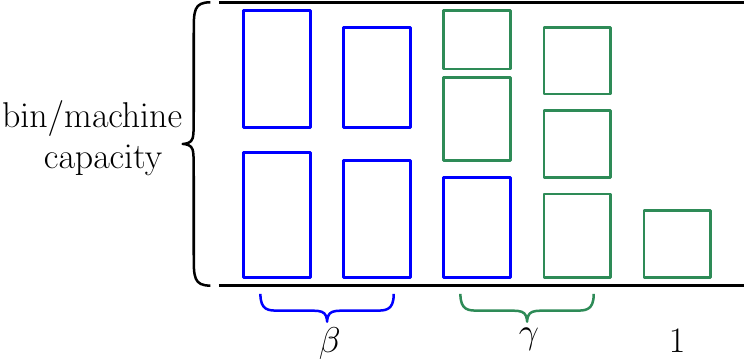}
  \caption{case $h' = 0$}
  \label{fig:ffd_h0}
\end{subfigure}%
\begin{subfigure}{.5\textwidth}
  \centering
  \includegraphics[width=\linewidth]{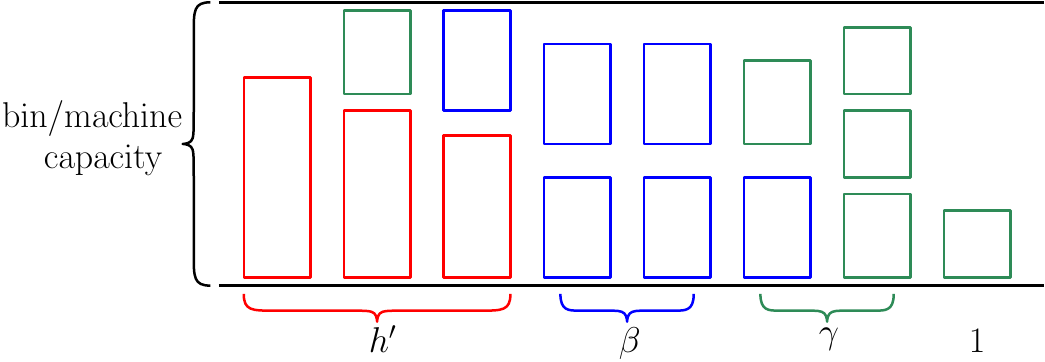}
  \caption{case $h' > 0$}
  \label{fig:ffd_hplus}
\end{subfigure}
\caption{FFD: each rectangle represents an item/job, and each column of rectangles represents the packing of a bin, and the bins are packed from left to right. Red, blue, and green items are huge, big-but-not-huge ones, and small jobs respectively. }
\label{fig:test}
\end{figure}

\section{Exact and Online Algorithms for the Huge-Tiny Case}
\label{s_hp}

Our exact algorithm $G$ and online algorithm $\hG$ for the Huge-Tiny case are very similar.
The two greedy algorithms are somewhat simple and natural, but the proofs of correctness and competitive ratio respectively
are fairly involved. In order to make the analysis a bit easier we make two simplifying assumptions in this section.
These assumptions will be removed in the Appendix \ref{s_pt}.

The first assumption is that we proceed timeslot by timeslot (and as a result the algorithm does not run in polynomial time).
In fact, as shown in the Appendix \ref{s_pt}, there are a polynomial number of certain ``events'' (such as the arrival of job or finishing all the
processing of a job) that trigger computation and all the other timeslots can be ignored.

The second  assumption is that
we split huge jobs into unit jobs; 
a huge job $j$ is \textit{split} into unit jobs $j_1,j_2,\ldots,j_{p(j)}$ such that $W_{j_i} = W_{j}$ for $i = 1,2,\ldots,p(j)$.
The basic idea on how to remove this assumption in an algorithm is as follows: 
the only danger when splitting a huge job into unit jobs is that two unit jobs
obtained from the same original huge job
are scheduled in the same timeslot.
In the offline setting, we can simply reschedule one of the two unit jobs
in a timeslot (guaranteed to exist from $p(j) \leq d(j) - a(j)$) 
where no unit jobs obtained from this original huge job are scheduled. 
This does not increase the objective of the schedule.
In the online setting, we would keep track of how many of the unit jobs coming from an original
huge job have been scheduled so far, and when the original job reaches its starting deadline:
this would create a batch as defined below and we could carefully schedule the unit jobs
in different timeslots.
A formal way of doing this appears in Appendix \ref{s_pt}, where both algorithms are shown
to each have a variant that does the same schedule and works in polynomial  time. 

After the splitting, $h$ is the total number of unit huge jobs.
This is the same $h$ as in Lower Bound (\ref{lb_tall_jobs}).
For each open timeslot $[t,t+1]$, if there is no huge job scheduled within $[t,t+1]$, then only one machine is used during $[t,t+1]$. 
If there is some huge job scheduled within $[t,t+1]$, the number of machines used is exactly the number of huge jobs scheduled within $[t,t+1]$. 
It follows that the total busy time is equal to the number of huge unit jobs (i.e. $h$) plus the number of open timeslots which do not have any huge job scheduled within. 
Consequently, the objective of minimizing the total busy time is equivalent to minimizing the number of open timeslots which do not have any huge job scheduled within. 

Here are the main ideas of both algorithms. We proceed timeslot by timeslot, 
with $Y$ and $\hY$ being the unscheduled huge jobs of $G$ and $\hG$ respectively.
{\bf First}, we check if a \emph{batch}
consisting of all the 
unscheduled huge jobs with deadline at most $\delta$
can cover all the timeslots from the current time to $\delta$.
If such a $\delta$ and such a batch exists, we make this schedule as, intuitively, there is nothing better one can do
with these huge jobs than spread them over as large an interval as possible. See Figure \ref{fig:bipartite} for an example. 
We also mark the timeslots where a huge job has already been scheduled.
The algorithms work the same except for the sets $Y$ and $\hY$, where $G$ has the
advantage of knowing the jobs arriving at further times (see the $1.5$ lower bound to see
how $\hG$ fails to find the optimum here). 
The check if a batch exists can be done by Claim \ref{claim:EDF_works} (also see Appendix \ref{s_pt}). 
Before moving to the next timeslot, we keep checking the existence of a batch and schedule its jobs if the batch exists. 

{\bf Second}, before moving to the next timeslot, if a timeslot already has a huge job, then all the available tiny jobs will be
scheduled in that timeslot. 
After that ({\bf third}), we look if any tiny job has {\bf reached its starting deadline}.
If so, we open the timeslot, schedule all the available tiny jobs, and also bring,
if possible, a huge job using EDF as criteria (the objective function suffers 
if we only schedule tiny jobs, and what better huge job to bring than the one whose
remaining window is smallest?). It is fairly intuitive that $G$ returns
an optimum solution;
our proof is however quite involved (as we could not find a simple proof). The proof for $G$ builds machinery that we also use for $\hG$.

Finding batches could be done by Bipartite Matching, or as described below.
For algorithms for Convex Bipartite Matching, as previous work \cite{Glover1967MaximumMI} has shown, Earliest-Deadline-First(EDF) works. 
The  following claim reproduces this result in the special case needed by our analysis.  
\begin{claim} \label{claim:EDF_works}
    Take any time $\tau\in \mathbb{Z}$ and any $\delta = \tau+1,\tau+2,\ldots$. 
    Let $J$ denote a set of huge unit jobs whose deadlines are at least $\tau+1$ and at most $\delta$. 
    If $J$ can cover all the timeslots during $[\tau,\delta]$, then EDF schedules a subset of $J$ at each timeslot during $[\tau,\delta]$ from left to right. 
\end{claim}

\begin{proof}
    Let $l$ denote $\delta - \tau$. 
    Suppose $j_{1}, j_{2}, \ldots, j_{l}$ are scheduled within $[\tau,\tau+1], [\tau+1,\tau+2], \ldots, [\tau+l-1, \tau+l]$ respectively. 
    Let $S$ denote the above scheduling, and precisely let $S=(s_1,s_2,\ldots,s_l)$ be a vector such that $s_i = (j_i, [\tau+i-1, \tau+i])$ for $i = 1,2,\ldots,l$. 
    Take any $k = 0,1,\ldots,l-1$. 
    We prove by induction on the value of $k$ that there exists a scheduling that uses EDF for the first $k$ timeslots and covers all the timeslots during $[\tau,\delta]$. 
    Assume that $S$ uses EDF for the first $k$ timeslots, and we find a scheduling that uses EDF for the first $k+1$ timeslots. 
    Suppose $j^*$ has the earliest deadline among all the jobs except $j_1,\ldots,j_k$ that can be scheduled within timeslot $[\tau+k, \tau+k+1]$, i.e., $j^*\in \argmin_{j\in J\setminus \{j_1,j_2,\ldots,j_k\}: \tau + k\in [a(j),d(j)-1]} d(j)$. 
    Note that $d(j^*) \leq d(j_{k+1})$ since $j_{k+1}$ is among the above set of jobs. 
    If $j_{k+1} = j^*$, then we are done. 
    It suffices to assume $j_{k+1} \neq j^*$. 
    If $j^*\notin \{j_{k+2},j_{k+3},\ldots,j_l\}$, then replace the $(k+1)$-th entry $s_{k+1}$ of $S$ by $(j^*,[\tau+k, \tau+k+1])$ and we are done. 
    If $j^* = j_i$ for some $i = k+2,k+3,\ldots,l$, then replace $s_{k+1}$ and $s_i$ by $(j^*,[\tau+k, \tau+k+1])$ and $(j_{k+1}, [\tau + i - 1, \tau + i])$. 
    Note that $j_{k+1}$ can be scheduled during $[\tau + i - 1, \tau + i]$ because $j^*$ is scheduled within timeslot $[\tau + i - 1, \tau + i]$ and $d(j_{k+1}) \geq d(j^*)$. 
\end{proof}

Note that the algorithm as described does not specify which $\delta$ to choose
if more than one are available, and this does not affect the correctness
arguments that follow. 
Appendix \ref{s_pt} does have a fixed  method for choosing $\delta$.

For each batch $B$ of huge jobs scheduled by the algorithms $G$  and $\hG$
we have an interval $I(B)$ of timeslots where the jobs of the batch are scheduled.
The union of such $I(B)$'s is partitioned into maximal intervals that we call
$H$-blocks for  $G$ and $\hH$-blocks for $\hG$.
We call both $H$-blocks and $\hH$-block \emph{blocks}.

A  \emph{critical} timeslot is one that is opened by the
algorithm but not covered by batch.
Therefore, a critical timeslot opens because of a tiny job
reaching its starting deadline.
All the tiny jobs reaching a starting deadline 
are \emph{critical} tiny jobs.
The critical timeslots are partitioned into
\emph{tiny-only-slots}   (where, as the name suggests, the algorithm does not schedule a huge job)
and \emph{covered-slots} (where a huge job is scheduled by EDF).

\begin{claim} \label{claim:block_batch}
    All the huge jobs are scheduled.
    Each huge job scheduled within a block is also in a batch. Each huge job scheduled outside any block covers some covered-slot. 
\end{claim}
\begin{proof}
    If a huge job $j$ with deadline $d$ is not scheduled before time $d-1$,
    then a batch with deadline $d$ can be found at time $d-1$, and
    therefore $j$ is scheduled.
    
    Observe that if a batch is scheduled to cover timeslot $[t,t+1]$,
    then all the available tiny jobs are scheduled during timeslot $[t,t+1]$.
    The algorithm checks for batches before handling tiny jobs at their starting deadline.
    Therefore, any huge job scheduled during timeslot $[t,t+1]$ cannot be via EDF. 
    Any huge job scheduled outside any block cannot be in any batch, and hence must  be scheduled by EDF. 
\end{proof}

    Below, $s^*$ denotes the number of tiny-only-slots in an optimum solution.
    For an interval $E$, define
    $p_E(j) = \max\{0, p(j) - L(W_j \setminus E)\}$,
    the minimum number of timeslots of $E$ that must be open to accommodate job $j$,
    where the function $L(\cdot)$ returns the number of timeslots in the function's argument.
    For a set of huge unit jobs $H$,
    define $maxcoverage(H)$, the maximum coverage of the set $H$, to be the maximum number of timeslots 
    that $H$ can cover.
    For example, in Figure \ref{fig:bipartite}, the maximum coverage of $j_2,j_3,j_4,j_5,j_6,j_7$ is $5$,
    where one maximum is achieved when 
    $j_2,j_3,j_4,j_6,j_7$ are scheduled during $[1,2]$, $[2,3]$, $[3,4]$, $[4,5]$, and $[7,8]$ respectively.

    It is easy to check that for any two sets of huge jobs $H$ and $H'$,  $maxcoverage(H \cup H') \leq maxcoverage(H) + maxcoverage(H')$, and  if $H \subset H'$, then 
    $maxcoverage(H) \leq maxcoverage(H')$.
    As an aside,  $maxcoverage(H)$ can be computed by Bipartite Matching algorithms
    (the bipartite graph will have the set of huge jobs as one
    vertex part, and the set of timeslots as the other vertex part),
    and in fact by the special case of Convex Bipartite Matching \cite{Glover1967MaximumMI,lipski1981efficient,GALLO198431,linear_time_convex}.

\begin{figure}
    \centering
    \includegraphics[width=0.5\linewidth]{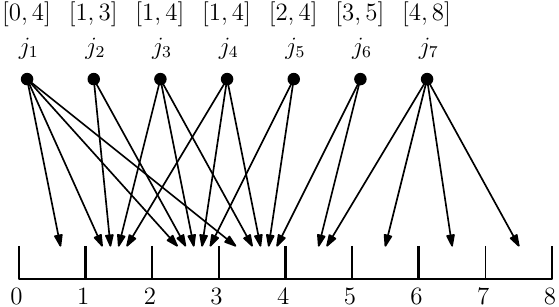}
    \caption{An example of batches with unit huge jobs, where we ignore tiny jobs for simplicity. 
    Suppose that the algorithm is at time $0$. 
    Each dot represents an unscheduled huge unit job that is also represented as an interval which is the window of the huge unit job. 
    Each arrow represents a possible assignment of a job into a timeslot. 
    For $G$, $Y$ is the set of all the huge unit jobs $\{j_1,j_2,\ldots,j_7\}$. 
    A batch can be found when $\delta = 4$ or $5$. 
    When $\delta = 4$, $\{j_1,j_2,j_3,j_4,j_5\}$ is a batch,
    where $j_1, j_2,j_3,j_4,j_5$ are scheduled during $[0,1]$, 
    $[1,2]$, $[2,3]$, $[3,4]$, and $[2,3]$ respectively. 
    When $\delta = 5$, $\{j_1,j_2,j_3,j_4,j_5,j_6\}$ is a batch, where $j_1,j_2,j_3,j_4,j_5,j_6$ are scheduled during $[0,1]$, $[1,2]$, $[2,3]$, $[3,4]$, $[2,3]$ and $[4,5]$ respectively. 
    For $\hG$, $\hY$ is $\{j_1\}$ because the rest of the jobs will arrive at or after time $1$,
    and hence no batch can be found because $j_1$ has a deadline at time $4$. }
    \label{fig:bipartite}
\end{figure}

    For any disjoint intervals $I_1, \ldots, I_u$ and  
    not necessarily distinct tiny jobs $j_1, \ldots, j_u$,
    we have:
\begin{equation}
\label{lb_hp'}
    s^* \geq \left(\sum_{l=1}^u p_{I_l}(j_l) \right)-
    maxcoverage(\{j \; | \; j \mbox{ huge and }   L(W_j \cap \left( \cup_l I_l\right)) \geq 1 \}). 
\end{equation}
\begin{proof} [Proof of (\ref{lb_hp'})]
    Take any optimal solution $S^*$.
    Suppose the number of open timeslots of $S^*$ during $\cup_l I_l$ is $K$. 
    Partition the $K$ open timeslots into the timeslots that are covered by some huge job and the timeslots that are not covered by any huge job, where we let $K_1$ and $K_2$ denote the cardinality of the two subsets respectively. 
    It follows that $K = K_1 + K_2$. 
    Note that $s^* \geq K_2$. 
    Note that the number of open timeslots of $S^*$ during $I_l$ is at least $p_{I_l}(j_l)$. 
    Since intervals $I_l$ are disjoint, we have $K \geq \sum^u_{l=1} p_{I_l}(j_l)$. 
    It remains to show $K_1 \leq maxcoverage(\{j \; | \; j \mbox{ huge and }  L( W_j \cap \left( \cup_l I_l\right) ) \geq 1 \})$. 
    Let $H$ denote the set of huge jobs scheduled during $\cup_l I_l$ by $S^*$. 
    We have $K_1 \leq maxcoverage(H)$ because jobs $H$ cover exactly $K_1$ timeslots in the scheduling $S^*$.
    We have $maxcoverage(H) \leq maxcoverage(\{j \; | \; j \mbox{ huge and }   L(W_j \cap \left( \cup_l I_l\right)) \geq 1 \})$, since the argument of $maxcoverage$ on the left is a subset of the argument of $maxcoverage$ on the right. 
\end{proof}

For each time instant $t\in \mathbb{Z}$, let $Y(t)$ denote the set of huge jobs available to $G$ to be scheduled at and after time $t$. 
Similarly, let $\hY(t)$ denote the set of huge jobs available to $\hG$ be scheduled at and after time $t$. 

\begin{claim}
\label{claim: any_batch_works}
    While running Algorithm $G$, if $Y(t_1)$ contains a set (which may or may not be chosen to be a batch by $G$)
    of huge jobs with deadlines at most $d_1$ that can cover all the timeslots during $[t_1,d_1]$,
    then each timeslot during $[t_1,d_1]$ is covered by some batch scheduled by $G$.  
\end{claim}

\begin{proof}
    We show the claim by contradiction.
    Assume that there exists some timeslot $[t^*,t^*+1]$ that is not covered by any batch and each timeslot within $[t_1,t^*]$ is covered by some batch, where $t^* = t_1,t_1+1,\ldots,d_1-1$. 
    Suppose a set of jobs $j_1,j_2,\ldots,j_l$ in $Y(t_1)$ can cover timeslots $[t_1,t_1+1], [t_1+1,t_1+2], \ldots, [d_1-1, d_1]$.
    We can take a subset and assume that $l = d_1 - t_1$ and, for $i \in \{1, 2, \ldots, l\}$, job $j_i$ can cover timeslot $[t_1 + i-1, t_1+i]$.
    Since no batch covers timeslot $[t^*,t^*+1]$, we have that at time $t^*$
    all the huge jobs that have been scheduled in batches have deadlines at most $t^*$. 
    Since each timeslot within $[t_1,t^*]$ is covered by some batch, Claim \ref{claim:block_batch} says that none of the timeslots within $[t_1,t^*]$ is critical and each of the huge jobs scheduled within $[t_1,t^*]$ is in some batch. 
    With $l' = t^* - t_1+1$, it follows that 
    $j_{l'}, \ldots, j_l$ are in $Y(t^*)$. 
    Note that $j_{l'}, \ldots, j_l$ have deadlines at most $d_1$ and can be scheduled within timeslots $[t^*,t^*+1], \ldots, [d_1-1, d_1]$ respectively. 
    By definition, algorithm $G$ should schedule a batch starting at time $t^*$, and hence we have a contradiction. 
\end{proof}

\begin{claim}
\label{claim: G_covereage}
    The maximum coverage of the set of huge jobs scheduled by $G$ in an $H$-block is the length of the block.
\end{claim}
\begin{proof}
    Take any $H$-block denoted by $[t_1,d_1]$. 
    We show it by contradiction. 
    Let $H^*$ denote the set of huge jobs scheduled in the taken $H$-block $[t_1,d_1]$. 
    Assume that $C^*$ is a schedule
    of $H^*$ on the set of timeslots $\{[t,t+1]\;|\;t\in T^*\}$
    such that no job in $H^*$ can be feasibly moved rightwards while $|T^*| \geq d_1 - t_1 + 1$. 
    Let $\hat{t} := \min T^*$ such that $[\hat{t},\hat{t}+1]$ is the first timeslot covered by $C^*$. 
    It follows from $|T^*| \geq d_1 - t_1 + 1$ that $\hat{t} \leq t_1 - 1$. 
    Let $\hat{d} := \min \{t\notin T^*: t\geq \hat{t}\}$ such that
    $[\hat{d},\hat{d}+1]$ is the first timeslot that is after time $\hat{t}$ and uncovered by $C^*$
    (it is possible that $\hat{d} = d_1$). 
    Note that all the jobs in $H^*$ have deadlines at most $d_1$. It follows that $\max T^* \leq d_1-1$. 
    Let $j_t$ denote the job in $H^*$ covering timeslot $[t,t+1]$ for each $t\in T^*$. 
    By the choice of the schedule $C^*$, we have $d(j_t) \leq \hat{d}$ for each $t = \hat{t}, \hat{t}+1,\ldots,\hat{d}-1$, otherwise some job would be moved rightwards to cover $[\hat{d}, \hat{d}+1]$.  
    Therefore, we find a nonempty set of jobs
    $\{j_t \; | \; t = \hat{t}, \hat{t}+1,\ldots,\hat{d}-1\}$ with deadlines
    at most $\hat{d}$ that can cover all the timeslots within
    $\left[ \hat{t},\hat{d} \right]$. 
    Furthermore, note that $\hat{d} \geq t_1+1$ since each job in
    $H^*$ has deadline at least $t_1+1$. 
    Recall $\hat{t} \leq t_1 - 1$.
    Note that $H^*\subseteq Y(t_1)\subseteq Y(\hat{t})$ since $\hat{t} \leq t_1 - 1$.
    Thus $Y(\hat{t})$ contains a set of huge jobs with deadlines at most $\hat{d}$
    that can cover all the timeslots during $[\hat{t},\hat{d}]$.
    Claim \ref{claim: any_batch_works} says timeslot $[t_1-1,t_1]$ is covered by some batch scheduled by $G$ (contradiction). 
\end{proof}

As an aside,  Claim \ref{claim: any_batch_works} holds for $\hG$, but in the proof above,  $\hY(t_1) \subseteq \hY(\hat{t})$  can fail
since the jobs arriving after time $\hat{t}$ but before $t_1$ are not known to the algorithm.

\begin{claim} \label{c_hH_block}
    Suppose a $\hH$-block starts at time $t_1$. 
    There can be at most $d_1 - t_1$ jobs in $\hY(t_1-1)$ with deadline at most $d_1$ for each $d_1 = t_1, t_1+1, \ldots$.
\end{claim}
\begin{proof} 
    Suppose $j_1,j_2,\ldots,j_{k}$ are the jobs in $\hY(t_1-1)$ with deadline at most $d_1$ such that $d(j_1) \leq d(j_2) \leq \cdots \leq d(j_{k}) \leq d_1$. 
    Since jobs $\hY(t_1-1)$ are available to be scheduled during timeslot $[t_1-1,t_1]$, we have $t_1 \leq d(j_1)$. 
    Assume $k \geq d_1-t_1+1$ for the sake of contradiction. 
    Let $x^*$ be the smallest number among $1,2,\ldots,k$ such that $d(j_{x^*}) = t_1-1+{x^*}$. 
    Note that $x^*$ exists, otherwise $d(j_x) \geq t_1 + x$ for each $x = 1,2,\ldots,d_1 - t_1 + 1$, and hence $d_1 + 1 \leq d(j_{d_1-t_1+1}) \leq d(j_k)$ would happen. 
    Since $x^*$ is the smallest such number, we have $d(j_x) \geq t_1 + x$ for each $x = 1,2,\ldots,x^*-1$ and $d(j_{x^*}) = t_1 - 1 + x^*$. 
    
    \textit{Method of EDF(see an argument based on Hall's theorem in Appendix \ref{a_ps}):} Consider scheduling jobs $j_1,j_2,\ldots,j_{x^*}$ within $[t_1-1,t_1 - 1 + x^*]$ by EDF. 
    It is clear that $j_1,j_2,\ldots,j_{x^*}$ are scheduled within $[t_1-1,t_1], [t_1, t_1+1], \ldots, [t_1-2+x^*, t_1 -1 + x^*]$ respectively by EDF (schedule the job with the smallest index when the ties happen). 
    
    It follows that algorithm $\hG$ schedules a batch starting from time $t_1-1$. 
    Since a $\hH$-block starts at time $t_1$, by definition, the timeslot $[t_1-1,t_1]$ is not covered by any batch. We have a contradiction on the scheduling during timeslot $[t_1-1,t_1]$. 
\end{proof}

\begin{claim}
\label{claim: hG_covereage}
    The maximum coverage of the set of huge jobs scheduled in an $\hH$-block is
    at most twice the length of the block.
\end{claim}

\begin{proof}
    Let $J$ denote the set of huge jobs scheduled in the $\hH$-block. 
    Suppose the $\hH$-block is $[t_1,d_1]$. 
    Let $J'$ be the set of huge jobs in $\hY(t_1-1)$ with deadline at most $d_1$. 
    By Claim \ref{claim:block_batch}, all the jobs $J$ are in batches and therefore,
    taking into account the maximality of the $\hH$-block,
    have deadlines at most $d_1$. 
    Since $\hY$ is obtained online, we have $t_1 \leq a(j)$ for each $j\in J\setminus J'$. 
    Now, take any scheduling $C$ of $J$. 
    In scheduling $C$, jobs from $J\setminus J'$ can only be scheduled in timeslots within the $\hH$-block $[t_1,d_1]$. 
    Claim \ref{c_hH_block} says $|J\cap J'|\leq |J'| \leq d_1 - t_1$. 
    It follows that the number of the timeslots outside (precisely, to the left of) the $\hH$-block that $C$ assigns to a job from $J$ is at most $d_1-t_1$. 
    Therefore, $C$ schedules jobs from $J$ in a number of timeslots at most twice the length of the $\hH$-block.
\end{proof}

We use an algorithm $\cA$ to construct the intervals and jobs from lower bound (\ref{lb_hp'})
(we tried and failed to construct these intervals with mathematical properties only).
We call these intervals, each with an attached job, $L$-intervals.
Algorithm  $\cA$ is used for both $G$ and $\hG$, with the only difference
that the algorithm $\cA$ for $G$ uses $H$-blocks and 
the algorithm $\cA$ for $\hG$ uses $\hH$-blocks.

Algorithm $\cA$ constructs (disjoint) $L$-intervals and attached jobs from right to left.
The $L$-intervals will either fully contain a block or be disjoint from it.
A block is called \emph{active}  if it is not contained in an $L$-interval,
and \emph{passive} otherwise. Note that active blocks may turn into passive blocks
as the algorithm $\cA$ constructs more $L$-intervals.

An \emph{active chain} connects a covered-slot to an active block in the following way.
An active chain that connects a covered-slot $[t^*,t^*+1]$ 
to an active block scheduled in the interval
$[t'_1,d_1]$ is a sequence of huge jobs
$j_q, j_{q-1}, \ldots, j_1$ such that $j_1$  is scheduled 
in timeslot $[t_1,t_1+1]$ within $[t'_1,d_1]$,
$j_x$ covers (by EDF) covered timeslot $[t_x, t_x+1]$ for $x = 2,3,\ldots,q$,
$a(j_x)\leq t_{x+1} < t_x$ for $x = 1,2,\ldots,q$, where $t_{q+1} := t^*$. 
Note that $t_2 < t'_1$.
See Figure \ref{fig:active_chain} for examples of active chains. 

\begin{figure} [h]
    \centering
    \includegraphics[width=0.8\linewidth]{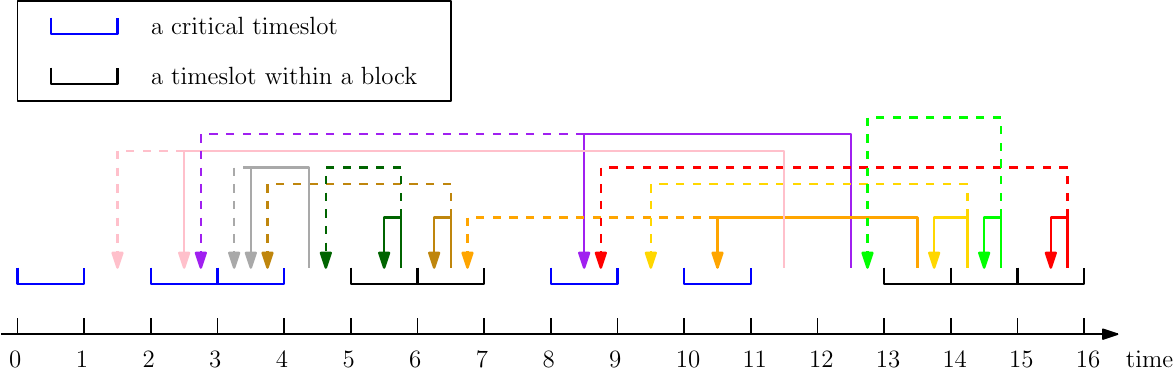}
    \caption{Active chains. Each color represents a huge job. For each huge job, the solid vertical line represents the last timeslot of the job's window,
    the solid arrow where the job is placed by the algorithm, and the dashed arrow the first timeslot of the job's window.
    Precisely, suppose huge jobs $\textcolor{pink}{2}, \textcolor{gray}{3}, \textcolor{teal}{5}, \textcolor{brown}{6}, \textcolor{violet}{8}, \textcolor{orange}{10}, \textcolor{yellow}{13}, \textcolor{green}{14}, \textcolor{red}{15}$ are scheduled within the open timeslots ${[2,3]}, {[3,4]}, {[5,6]}, {[6,7]}, {[8,9]}, {[10,11]}, {[13,14]}, {[14,15]}, {[15,16]}$ respectively. 
    We have $W_{\textcolor{pink}{2}} = [1,12]$, $W_{\textcolor{gray}{3}} = [3,5]$, $W_{\textcolor{teal}{5}} = [4,6]$, $W_{\textcolor{brown}{6}} = [3,7]$, $W_{\textcolor{violet}{8}} = [2,13]$, $W_{\textcolor{orange}{10}} = [6,14]$, $W_{\textcolor{yellow}{13}} = [9,15]$, $W_{\textcolor{green}{14}} = [12,15]$, and $W_{\textcolor{red}{15}} = [8,16]$. 
    $[13,16]$ is an active block. 
    Jobs $j_1 = \textcolor{red}{15}, j_2 = \textcolor{violet}{8}$ constitute an active chain that connects the covered-slot $[3,4]$ to the active block $[13,16]$.
    $[5,7]$ is an active block. 
    Job $j_1 = \textcolor{brown}{6}$ constitutes an active chain that connects the covered-slot $[3,4]$ to the active block $[5,7]$. 
    }
    \label{fig:active_chain}
\end{figure}

A covered-slot that is the endpoint of an active chain is called \emph{deceptive},
the covered-slots that are not deceptive are called \emph{honest}.
Note that deceptive covered-slots may become honest as algorithm $\cA$ progresses,
as there  may be fewer active blocks.

The algorithm $\cA$ proceeds from right to left and finds the first (critical) timeslot $[t_r, t_r+1]$
that is tiny-only or honest. $\cA$ constructs an $L$-interval $I$ with right endpoint $t_r+1$.
Then the algorithm $\cA$ selects and attaches to $I$ a
critical job  $j$ that has a starting deadline at $t_r$, which means that all the timeslots
between $t_r$ and $d(j)$ have the job $j$ assigned to them. 
If the arrival time of $j$, $a(j)$ falls in a block $B$, then
$I$ starts
right where $B$ ends; otherwise $I$ starts at the arrival time of $j$.
After $I$ is constructed, we look for the next 
(going right-to-left, and starting strictly before the left endpoint of $I$) 
critical timeslot that is tiny-only or honest (note that the set of active blocks might
have changed, as a block may end up in $I$, and therefore some deceptive covered timeslots may become honest).
And the process continues. 

\begin{claim} \label{claim:cA_results}
        The $L$-intervals created by $\cA$ are disjoint.
        After $\cA$ finishes executing, any tiny-only-slot or any honest covered-slot is inside an $L$-interval, and  every passive 
        block is inside $\left( \cup_l I_l\right)$ and every active block is  outside $\left( \cup_l I_l\right)$. 
\end{claim}

\begin{proof} 
    It is straightforward from the description of $\cA$ that the $L$-intervals created are disjoint.

    If a tiny-only-slot or honest covered slot is to the left of the last $L$-interval created by $\cA$ or $\cA$ has not created $L$-intervals yet,
    then a new $L$-interval will be created. 
    From the description of $\cA$ when an $L$-interval $I$ is created, there cannot be tiny-only-slots or honest covered slots between this interval
    and the $L$-interval to the right of $I$, or to the right of $I$ if $I$ is the first $L$-interval created by $\cA$.
    Moreover, once $I$ is created, the deceptive covered slots after the right endpoint of $I$ do not become honest, since all the active blocks 
    that become passive from now on are in $I$ or to the left of $I$. Thus 
        after $\cA$ finishes executing, any tiny-only-slot or any honest covered-slot is inside an $L$-interval.
    
    For the last statement, it suffices to ensure that $t^*\in [t_1+1,d_1-1]$ cannot happen where $t^*$ is an endpoint time of a $L$-interval and $[t_1,d_1]$ is a block.  
    Note that the rightmost timeslot of a $L$-interval is a critical timeslot, which covers the case when $t^*$ is a right endpoint time. 
    The choice of the leftmost timeslot of a $L$-interval in algorithm $\cA$ covers the case when $t^*$ is a left endpoint time. 
\end{proof}

\begin{claim} \label{claim:active_chain_1}
    Suppose $j_q, j_{q-1}, \ldots, j_1$ is an active chain whose right end 
    is active block $[t'_1,d_1]$. 
    We have $d(j_x) \leq d_1$ for each $x = 1,2,\ldots,q$. 
\end{claim}

\begin{proof}
    Note that $j_1$ is scheduled within the active block $[t'_1,d_1]$. 
    Claim \ref{claim:block_batch} says $j_1$ is in some batch $B$ scheduled by algorithm $G$ (or $\hG$). 
    By definition of a batch, we have $d(j_1) \leq \max I(B) \leq d_1$ where $\max I(B)$ is the right endpoint of the interval $I(B)$. 
    Next, take any $x = 1,2,\ldots,q-1$. 
    We show that $d(j_x) \leq d_1$ implies $d(j_{x+1}) \leq d_1$. 
    Suppose $j_x$ covers timeslot $[t_x,t_x+1]$ and $j_{x+1}$ covers timeslot $[t_{x+1}, t_{x+1}+1]$. 
    By the definition of active chain, we have $t_{x+1} < t_x$ and $a(j_x) \leq t_{x+1}$. 
    It follows that job $j_x$ is available to be scheduled within timeslot $[t_{x+1}, t_{x+1}+1]$ at time $t_{x+1}$, because $a(j_x) \leq t_{x+1} < t_{x+1}+1 \leq t_x < t_x + 1 \leq d(j_x)$ and the fact that $j_x$ is scheduled within $[t_{x},t_{x}+1]$ implies $j_x$ has not been scheduled at time $t_{x+1}$. 
    It follows from the definition of EDF and the induction hypothesis that $d(j_{x+1}) \leq d(j_x) \leq d_1$. 
\end{proof}

\begin{claim} \label{claim:active_chain_2}
    At the end of $\cA$, if an active chain connects covered-slot $[t^*,t^*+1]$ to active block $[t'_1,d_1]$, then for any open timeslot $[t,t+1]$ such that $t^* \leq t \leq t'_1-1$, either $[t,t+1]$ is in some active block, or, deceptive covered. 
\end{claim}
\begin{proof}
    Suppose $j_q, j_{q-1}, \ldots, j_1$ connects the covered-slot $[t^*,t^*+1]$ to active block $[t'_1,d_1]$. 
    Suppose $j_x$ covers timeslot $[t_x,t_x+1]$ for each $x = 1,2,\ldots,q$ and $t_{q+1}:=t^*$. 
    By the definition of active chain, we have
    $t_{q+1} < t_q < \cdots < t_2 < t_1$, and also $t_2 < t'_1 \leq t_1 < d_1$.
    Take $q^*$ such that $t_{q^*+1} \leq t < t_{q^*}$. 
    It is easy to see that $q^*\in \{1,2,\ldots,q\}$. 
    It follows from the definition of active chain that $a(j_{q^*}) \leq t_{q^*+1} \leq t < t_{q^*} < t_{q^*}+1 \leq d(j_{q^*})$
    and hence job $j_{q^*}$ is available to be scheduled at time $t$. 
    It follows that timeslot $[t,t+1]$ is not tiny-only. 
    Case 1: timeslot $[t,t+1]$ is not within any block. 
    Then, $[t,t+1]$ is a covered-slot. 
    It follows that timeslot $[t,t+1]$ is deceptive because jobs $j_{q^*}, \ldots, j_2,j_1$ 
    form  an active chain that connects $[t,t+1]$ to the active block $[t'_1,d_1]$.     
    Case 2: timeslot $[t,t+1]$ is within a block. Suppose (for a contradiction) that this block is passive.
    Then this block is contained in an $L$-interval $I$. 
    As the block $[t'_1,d_1]$  is active, it must be to the right of $I$.
    The rightmost timeslot of $I$ is a tiny-only-slot or an honest covered timeslot;
    however Case 1 shows that this timeslot (which is not in a block) is a deceptive covered slot, a contradiction.
\end{proof}

\begin{claim}
\label{c_pIj}
    Given an $L$-interval $I$ with attached job $j$,
    $p_I(j)$ is the number of timeslots that $j$ is assigned (by $G$ or $\hG$) to be processed
    inside $I$ and is equal to the number of open timeslots during $I$.
\end{claim}
\begin{proof} 
    Suppose $I = [I^-, I^+]$. 
    First, we show that $G$ (or $\hG$) schedules job $j$ at each timeslot during $W_j\setminus I = [a(j), I^-] \cup [I^+, d(j)]$. 
    Looking at $I^+$, by definition, $j$ reaches its starting deadline at timeslot $[I^+-1, I^+]$. 
    It follows that $j$ is scheduled at each timeslot during $[I^+, d(j)]$. 
    It suffices to consider the case when $a(j) < I^-$. 
    By the choice of $I^-$, there is a block containing $[a(j), I^-]$. 
    It follows from the definition of Algorithm $G$ (or $\hG$) and the fact that $j$'s processing is not completed at time $I^-$ that $j$ is scheduled at each timeslot during $[a(j), I^-]$. 
    Second, we show that any open timeslot during $I$ has $j$ scheduled in it. 
    This follows from the definition of Algorithm $G$ (or $\hG$)
    and the fact that $j$'s processing is not completed at time $I^+-1$. 
\end{proof}

Let $I_1, \ldots, I_u$ be the $L$-intervals, from right to left, constructed by Algorithm $\cA$ from the schedule of $G$ (or $\hG$). 
Let $H^*$ denote the set $\{j \; | \; j \mbox{ huge and }   L(W_j \cap \left( \cup_{l = 1,2,\ldots,u} I_l\right)) \geq 1 \}$. 
\begin{claim}
\label{c_inside}
    In the scheduling of $G$ (or $\hG$), any huge job from the set $H^*$ is scheduled within $\left( \cup_{l = 1,2,\ldots,u} I_l\right)$.
\end{claim}
\begin{proof} 

    \begin{figure}
        \centering
        \includegraphics[width=0.8\linewidth]{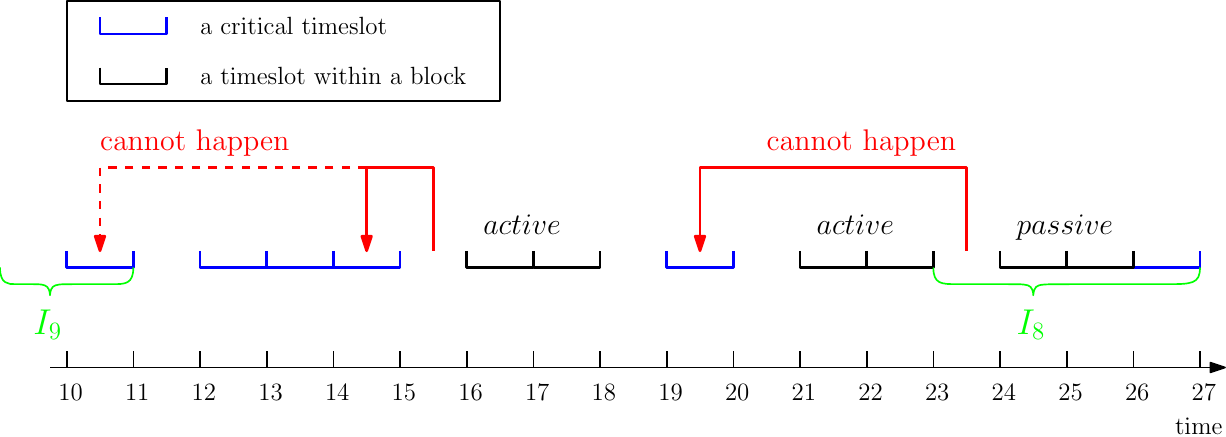}
        \caption{Consider a huge job that covers timeslot $[14,15]$ and Claim \ref{c_inside} says the huge job cannot arrive before $I_9$ ends. Consider a huge job that covers timeslot $[19,20]$ and Claim \ref{c_inside} says the huge job cannot have deadline after $I_8$ starts. }
        \label{fig:active_chain_L_intervals}
    \end{figure}

    Take any huge job ${j^*}\in H^*$.
    For the sake of contradiction, assume ${j^*}$ is scheduled within timeslot $[t^*,t^*+1]$ outside $\left( \cup_l I_l\right)$, say $\max I_{l+1} \leq t^* < \min I_{l}$, where $\min I_0:=\infty$ and $\max I_{u+1} = -\infty$. 
    Claim \ref{claim:cA_results} says $[t^*,t^*+1]$ is in an active block or a deceptive covered-slot. 

    Consider the former case. 
    Suppose $j^*$ is scheduled within an active block $[t_1,d_1]$. 
    It follows that $j^*$ is scheduled in a batch,
    and from the maximality of the block, we obtain $d(j^*) \leq d_1$. 
    On the other hand, we have $d_1 \leq \min I_l$ because the active block is disjoint from any $L$-interval (from Claim \ref{claim:cA_results}). 
    It follows that $W_{j^*}$ cannot intersect $I_l$ (if exists). 
    Suppose (for a contradiction) $W_{j^*}$ intersects with $I_{l+1}$. 
    Note that the rightmost timeslot within $I_{l+1}$ (if $I_{l+1}$  exists), say $[t^-, t^-+1]$, is tiny-only or honest covered, by the definition of $\cA$. 
    However,  $[t^-, t^-+1]$ cannot be tiny-only,
    since $j^*$ was available to cover it.
    And if $[t^-, t^-+1]$  is covered,
    the single job $j^*$ is an active chain that connects the critical timeslot $[t^-, t^-+1]$
    to the active block $[t_1,d_1]$, which says $[t^-, t^-+1]$ is a deceptive covered-slot (contradiction). 

    Consider the latter case when $[t^*,t^*+1]$ is a deceptive covered-slot. 
    See Figure \ref{fig:active_chain_L_intervals} for an example.  
    Suppose $j_q, j_{q-1}, \ldots, j_1$ form the active chain that connects covered-slot $[t^*,t^*+1]$ to an active block $[t_1,d_1]$. 

    Here, we show that $W_{j^*}$ does not intersect $I_{l+1}$ (if exists). 
    Let timeslot $[t^-,t^-+1]$ denote the rightmost timeslot within $I_{l+1}$. 
    Similar to the former case, the definition of algorithm $\cA$ says timeslot $[t^-,t^-+1]$ is either tiny-only or honest covered. 
    Assume the intersection, i.e., $a(j^*) \leq t^-$ for a contradiction. 
    Since $j^*$ is available to be scheduled at time $t^-$, timeslot $[t^-, t^-+1]$ must be covered. 
    It follows that $j^*,j_q, j_{q-1}, \ldots, j_1$ form an active chain that connects timeslot $[t^-,t^-+1]$ to an active block $[t_1,d_1]$ (one may check the conditions of being an active chain). 
    The definition of an active chain says timeslot $[t^-,t^-+1]$ is deceptive covered (contradiction). 

    Next, we show $W_{j^*}$ does not intersect $I_{l}$ (if exists). 
    Let timeslot $[t^+,t^++1]$ denote the rightmost timeslot within $I_{l}$, which is either tiny-only or honest covered. 
    We have $d_1 \leq \min I_l$, otherwise Claim \ref{claim:cA_results} says $\max I_l \leq t_1$ and then Claim \ref{claim:active_chain_2} says timeslot $[t^+,t^++1]$ would be deceptive (a contradiction). 
    Claim \ref{claim:active_chain_1} gives $d(j_q) \leq d_1$. 
    However, $d(j^*) \leq d(j_q)$ since $j^*$ was used to cover the
    timeslot $[t^*,t^* + 1]$
    that $j_q$ can also cover, and therefore $d(j^*) \leq d_1$.
    Hence $W_{j^*}$ does not intersect $I_l$.     
\end{proof}

\begin{claim} \label{claim:alg_H^*}
    In the scheduling of $G$ (or $\hG$), $H^*$ covers exactly
    $\left( \sum_x b_x \right) + |H^*\setminus \left( \cup_x H_x\right)|$ 
    (number of) timeslots, where $H_x$ is the set of huge jobs scheduled within 
    the $x$-th passive block and $b_x$ is the length of the block. 
\end{claim}
\begin{proof} 
    Note that every huge job in every $H_x$ is in $H^*$ because passive blocks are included in
    $L$-intervals.
    Claim \ref{c_inside} says that each huge job in $H^*\setminus \left( \cup_x H_x\right)$ is scheduled by EDF at some critical timeslot during $\cup_{l = 1,2,\ldots,u} I_l$. 
    Since each critical timeslot accommodates at most one huge job in our algorithm's scheduling, it follows that $H^*$ covers exactly $\left( \sum_x b_x \right) + |H^* \setminus \left( \cup_x H_x\right)|$ number of timeslots in the algorithm's scheduling, where the first term and second term correspond to the passive blocks and the critical timeslots respectively. 

\end{proof}

\begin{theorem}
\label{thm: G_opt}
    $G$ computes an optimum schedule.
\end{theorem}
\begin{proof}
    Let $\mu$ denote the number of tiny-only-slots in $G$'s scheduling. 
    The goal is to show $\mu \leq s^*$. 
    With the help of the lower bound (\ref{lb_hp'}),
    it suffices to show that $\mu \leq \left(\sum_{l=1}^u p_{I_l}(j_l) \right)- maxcoverage(H^*)$. 
    Let $\eta$ denote the number of timeslots $H^*$ covers (all within $\left( \cup_{l = 1,2,\ldots,u} I_l\right)$, by Claim \ref{c_inside}) in $G$'s scheduling. 
    By Claim \ref{claim:alg_H^*}, we have $\eta = \left( \sum_x b_x \right) + |H^*\setminus \left( \cup_x H_x\right)|$. 
    By Claim \ref{claim:cA_results} and \ref{c_pIj}, we have
    $\mu + \eta = \sum_{l = 1}^u p_{I_l}(j_l)$. 
    With the help of the lower bound (\ref{lb_hp'}),
    it suffices to show that $\eta \geq maxcoverage(H^*)$. 
    It suffices to show $maxcoverage(H^*) \leq \left( \sum_x b_x \right) + |H^*\setminus \left( \cup_x H_x\right)|$. 
    Take any schedule $C$ of $H^*$. 
    Claim \ref{claim: G_covereage} says $H_x$ covers at most $b_x$ timeslots in $C$. 
    It follows that $H^*$ covers at most $\left( \sum_x b_x \right) + |H^*\setminus \left( \cup_x H_x\right)|$ timeslots in $C$. 
\end{proof}

\begin{theorem}
\label{thm: hG}
    $\hG$ produces a schedule with at most $s^* + h/2$ tiny-only-slots.
\end{theorem}

\begin{proof}
We use the same notations as in the proof of Theorem \ref{thm: G_opt}. 
Claim \ref{claim: hG_covereage} says $maxcoverage(H_x) \leq 2\cdot b_x$ for each $x$. 
Note that we have $maxcoverage(H_x) \leq |H_x|$. 
It follows that $maxcoverage(H_x) \leq b_x + \frac{1}{2}|H_x|$. 
Let $\eta$ denote the number of timeslots $H^*$ covers
(all within $\left( \cup_{l = 1,2,\ldots,u} I_l\right)$,
by Claim \ref{c_inside}) in $\hG$'s scheduling. 
It follows that $maxcoverage(H^*) \leq \left( \sum_x maxcoverage(H_x) \right) + |H^*\setminus \left( \cup_x H_x\right)| \leq \left( \sum_x b_x + \frac{1}{2}|H_x| \right) + |H^*\setminus \left( \cup_x H_x\right)| \leq \eta + \frac{h}{2}$. 
In conclusion, we have 
$\mu = \left(\sum_l p_{I_l}(j_l) \right) - \eta \leq \left(\sum_l p_{I_l}(j_l) \right) - maxcoverage(H^*) + \frac{h}{2} \leq s^* + \frac{h}{2}$. 
\end{proof}

Note that the objective in terms of busy time is at most
$h + s^* + h/2 \leq (3/2)(h + s^*)$ and thus Algorithm $\hG$ matches the lower bound
of Theorem \ref{t_32}.

\section{Offline and Online Algorithms}
\label{s_offline&online}

 Using Algorithm $G$, solve the Huge-Tiny problem exactly, 
 where all the non-huge jobs are treated as tiny.
 The objective of the solution is $h + s^*$,
 where $s^*$ timeslots are open where no huge job is scheduled
 and recall that $h$ is the number of unit huge jobs.
 Note that this $h$ is the same as the $h$ from Lower Bound (\ref{lb_tall_jobs}).
 Run FFD in each timeslot, and let $a(i)$ be the number of resulting bins for timeslot $[i,i+1]$.
 Let $a = \sum_i a(i)$, be the total
 number of bins used, which is the objective of the busy-time problem.
 From Lemma \ref{lem: ffd}, by adding over the timeslots,
 we obtain $a \leq h + s^* + \frac{1}{18} b + \frac43 v$,
 where, as defined in Section \ref{s_prel},
 $b = \sum_{j \in  J_B \setminus J_H} p(j)$,
 and $v$ is the total volume. 
 As $\opt \geq h + s^*$, $\opt \geq \frac12 b$, and $\opt \geq v$,
 we obtain an approximation ratio of $1 + \frac19 + \frac43 = \frac{22}{9} \leq 2.444445$.

Using Online Algorithm $\hG$, obtain an approximate solution to the Huge-Tiny problem, where all the non-huge jobs are treated as tiny.
This approximate solution  has at most $s^* + h/2$ timeslots with only tiny jobs assigned,
where (as above) $s^*$  is the 
minimum number of timeslots  with  only tiny jobs assigned
in an optimum solution to the Huge-Tiny instance, 
 and recall that $h$ is the number of unit huge jobs.
 Run FFD in each timeslot, and let $a(i)$ be the number of resulting bins for timeslot $[i,i+1]$.
 Let $a = \sum_i a(i)$, be the total
 number of bins used, which is the objective of the busy-time problem.
 From Lemma \ref{lem: ffd}, by adding over the timeslots,
 we obtain $a \leq \frac{h}{2} + s^* + \frac{h}{2} + \frac32 v$, 
 where $v$ is the total volume. As $\opt \geq h + s^*$,
and $\opt \geq v$,
 we obtain a competitive ratio of $2.5$.

\section{Example for offline and online algorithms}
\label{s_242}
Consider the sequence defined by $c_1 = 3$ and
$c_i = c_{i-1}(1 + c_{i-1})$ for $i > 1$, and the series
$\sigma = \sum_{i \geq 1} 1/c_i \geq 0.423$. The family of examples below 
shows that the approximation and competitive ratios of 
our algorithms from Section \ref{s_offline&online} are at least
$2 + \sigma$.  For the offline algorithm,  this is very close
to the approximation ratio proven  in Appendix \ref{s_chatgpt}.
Let $1 \geq \eps > 0$; next we give an instance where
the objective of the output of our algorithms (based on 
$G$ and $\hG$) is at least $(2 + \sigma - \eps) \opt$.

Choose a large $k$ and a very small 
$\eps' < \frac{\eps}{2 c_k c_{k+1}}$,
and a very large $q$ that is a multiple of $c_k$.
The instance has no huge jobs.
For each time $t \in \{0,1, \ldots, k+q-1 \}$, one 
rigid unit job $\hat{j}_t$ with height $\eps'$ is released at time $t$.
$2q$ flexible unit jobs with height $1/3 + \eps'$ are released
at time  $0$ and all have deadline at time $k+q$. 
For $i =  1, 2, \ldots, k$,  $q$ 
flexible unit jobs with height $ 1/(c_i+1) + \eps'$ are released at time  $i$ and all have deadline at time $k+q$. 
Thus the number of jobs $n = k+q + 2q + k q = (3+k)q + k$.

The  algorithm based on $G$ (or $\hG$) opens $q$ bins during
timeslot $[0,1]$ for accommodating the rigid unit job with
$\eps'$ height and the $2q$ flexible
unit jobs with height $1/3 + \eps'$. 
For $i \in \{1,2, \ldots, k \}$, 
the algorithm based on $G$ (or $\hG$) opens $q/c_i$ bins in
timeslot $[i,i+1]$ for accommodating the rigid unit job with
$\eps'$ height and the $q$ flexible
unit jobs with height $1/(c_i+1) + \eps'$  that have
just been released. 
The algorithm based on $G$ (or $\hG$) opens one bin in
each of the $q-1$ timeslots in $[k+1,k+q]$. 
The total number of bins open is
$q + q \sum_{i=1}^k \frac{1}{c_i} + q-1$.

Meanwhile, the optimal scheduling opens one bin in 
each of the $k+q$ timeslots within $[0,k+q]$, 
as described below.
For $t \in \{0, 1, \dots, k-1\}$,
the bin during $[t,t+1]$ accommodates $\hat{j}_t$.
For $t \in \{k, k+1, \ldots, k+q-1 \}$, the bin during $[t,t+1]$ 
accommodates $\hat{j}_t$, a pair of flexible unit jobs with
height $1/3 + \eps'$ and, for $i \in \{1, 2, \ldots, k\}$,
one flexible unit job with height $1/(1+c_i) + \eps'$.
To see why this is feasible (and explain why we chose 
the sequence $c_i$), notice that
after we put in the bin the rigid job and the two big 
but not huge jobs, the space left in the bin
is $1/c_1 - 3 \eps'$.
The formula $c_{i+1} =  c_i (c_i + 1)$ allows us to 
maintain that, after putting in the job of size
$1/(c_i+1) + \eps'$, the space left in the bin is
$1/c_{i+1} - (i+3) \eps'$.

The approximation ratio of the algorithm based on $G$ 
and the competitive ratio of the algorithm based on $\hG$
are at least 
$\frac{  q \left( 2 +  \sum_{i=1}^k \frac{1}{c_i} \right) -1 }{k+q}$,
which exceeds $2+\sigma - \eps$ as $k$ grows and $q$ is
chosen to be much larger than $k$.

\bigskip



{\bf Acknowledgments:} We thank Varun Gupta and Samir Khuller  for many useful discussions related to this project.
ChatGPT and Claude were used for verifying these results, resulting in many minor corrections.

Appendix \ref{s_chatgpt} was obtained and written by ChatGPT Pro Lite in October 2026, based on prompts from the authors. It was carefully verified and minimally changed.

\bibliography{references}

\appendix

\section{$3$-approximation algorithm}
\label{s_3}

Here we prove the $3$-approximation ratio of the algorithm that solves Span Minimization in the first step and then uses AnyFit for packing the unit jobs into machines for each timeslot. 
Take any open timeslot. 
If the number of machines during the taken timeslot is $1$,
then use Lower Bound (\ref{lb_span})
to cover the busy time of $1$. 
In the case that the number of machines during the taken timeslot is greater than $1$, by the definition of AnyFit, it is impossible to have the sum of a pair of bin levels upper bounded by $1$, where the bin level of a bin is the total heights of jobs accommodated within the bin. 
It follows that either all the bin levels are above $1/2$, or, some bin level is at most $1/2$ but the average bin level is still above $1/2$ after matching it with another bin level. 
It follows that the total height of jobs scheduled during the taken timeslot is at least half of the number of bins. Lower bound (\ref{lb_volume}) is used to cover the busy time during the taken timeslot which is the number of bins opened by AnyFit. 
Eventually, the total busy time (total number of bins) is at most $|Q| + 2\cdot v \leq 3\cdot \opt$.

\section{Hall's theorem}
\label{a_ps}

\begin{proof} [Proof of Claim \ref{c_hH_block} using Hall's theorem]
    Looking at the bipartite graph between $j_1,j_2,\ldots,j_{x^*}$ and timeslots within $[t_1-1, t_1 - 1 +x^*]$. Take any subset of $j_1,j_2,\ldots,j_{x^*}$, say $j_{x_1}, j_{x_2}, \ldots, j_{x_k}$ with $x_1 < x_2 < \ldots <x_k$. Note that $j_{x_k}$ can be scheduled at each timeslot during $[t_1-1, t_1 - 1 + x_k]$. We have the number of neighbors of the vertices $j_{x_1}, j_{x_2}, \ldots, j_{x_k}$ is at least $x_k \geq k$. 
    By Hall's theorem, the jobs $j_1,j_2,\ldots,j_{x^*}$ can be matched to the timeslots during $[t_1-1, t_1 - 1 +x^*]$. 
\end{proof}

\input{running_time_section}

\input{a_tighter_analysis}

\section{Lower bound 3/2 for online deterministic algorithms, in the uniform height setting when $g\geq 4$}
\label{s_lb}

Here we assume that the machine capacity is 1 and every
job has height $1/g$.
Release a group of jobs $J_1$ at time $0$ such that $p(j) = 1$ and the window $W_j = [0,2]$ for each $j\in J_1$ and $|J_1| = g-2$ (also $\sigma(j) = 1/g$ uniformly). At the same time, release a job $j^1$ such that $p(j^1) = 2$ and $W_{j^1} = [0,3]$.
Now, it depends on whether algorithm schedules any $j\in J_1$ during
timeslot $[0,1]$. 

If so (case 1), release a job $j^2$ at time $1$ such that $p(j^2) = 2$ and $W_{j^2} = [1,3]$. 

If not (case 2), release group of jobs $J_2$ at time $1$ such that $p(j) = 1$ and the window $W_j = [1,2]$ for each $j\in J_2$ and $|J_2| = g-2$. 

For the analysis, in case 1, the algorithm cost is
at least $3$ because each timeslot has opened one machine. 
In the optimal scheduling, one machine is opened during timeslot $[1,2]$ which accommodates $g-2$ units from $J_1$ and $1$ unit from $j^1$ and $1$ unit from $j^2$, and one machine is opened during timeslot $[2,3]$ which accommodates $j^1$ and $j^2$. 
In case 1, we have ratio of $3/2$. 
In case 2, we have two subcases. 
In the first subcase, which can only happen when $g=4$,
the algorithm opens just one machine at timeslot $[1,2]$;
in this case, the jobs of $J_1 \cup J_2$ fill this machine, and
another two machines must be opened to accommodate $j_1$.
In the second subcase, the algorithm opens at least two machines at timeslot $[1,2]$;
in this case another machine must be opened to accommodate $j_1$,
during either $[0,1]$ or $[2,3]$. 
Thus, the algorithm's cost is also at least $3$. 
The optimal scheduling will open one machine during $[0,1]$ for accommodating $J_1$ and $j^1$ and one machine during $[1,2]$ for accommodating $J_2$ and $j^1$. 
The ratio is also $3/2$. 

\end{document}

%% file: running_time_section.tex
\section{Polynomial-Time Variants of the Algorithms}
\label{s_pt}

The outline of the algorithms is given in Section \ref{s_hp}.
We will argue in this section that the timeline will be split into $O(n)$ intervals, which we call \emph{segments},
such that exactly the same set of jobs $J(I)$ are scheduled in each timeslot of segment $I$. 
For each $J(I)$, we run FFD in time $O(n \log n)$
\cite{Wang2001}, and thus the final complete schedule is obtained in $O(n^2 \log n)$,
assuming we find all these segments $I$ within this time bound. This is what we achieve below.

For each segment, we keep a (sorted) list of the jobs assigned to this segment. Sometimes the algorithm splits segments; 
in this case we simply copy these sorted lists (maybe something smarter can be done, but it will not improve our overall running time).
There will be $O(n)$ times when we split segments, and each split can be done in $O(n)$.
When we schedule a job $j$ in an interval $I$, we identify in $O(n)$ all the segments that intersect $I$ and add job $j$ to the lists of these
segments. Also, two of the segments can intersect $I$ partially; these two segments are split as discussed earlier.

To find these segments (and the jobs scheduled in it), we use a variant of the 
 algorithms from Section \ref{s_hp}.
 We proceed on the timeline from left to right, but we do not go timeslot-by-timeslot anymore.
 A job is called {\em active} if its remaining processing time is positive.
 At time $t$ we keep a set $P(t)$ of active tiny jobs with their associated \emph{remaining} processing times and deadlines,
 and a set $H(t)$ of active huge jobs with their associated \emph{remaining} processing times, arrival times, and deadlines.
 All the jobs of $P(t)$ have arrival time at $t$ or before. For $\hG$, the algorithm uses $H(t)$ to store the active huge jobs
 with arrival time $t$ or earlier.  $G$ however keeps in $H(t)$ all the active huge jobs.
 It will be an invariant of the algorithm that the schedule we have done before time $t$ can be completed to a full schedule.

 The algorithm use procedure LeftmostBatch($t$)  which, based on $H(t)$,
 returns the smallest $t' \geq t$ for which a batch exists that
 covers timeslots starting with $[t',t'+1]$.

 If we (the algorithm)  schedule all the remaining processing time  of a job $j$, we say that we \emph{fulfill} job $j$.
At time $t$, the algorithm seeks the next \emph{event},
which may change the two sets $P()$ and $H()$ and possibly mark the end of a segment and the beginning of another segment.
These events happen at time $\hat{t} \geq t$ if one of the following holds (the algorithm checks them in this order):
    \begin{enumerate}
        \item a tiny job arrives at time $\hat{t}$; 
        \item a huge job arrives at time $\hat{t}$; 
        \item a tiny job is fulfilled at time $\hat{t}$;  (the last time where the job is scheduled being $[\hat{t}-1,\hat{t}]$)
        \item a huge job is fulfilled at time $\hat{t}$;  (the last time where the job is scheduled being $[\hat{t}-1,\hat{t}]$)
        \item a batch starts at time $\hat{t}$; 
        \item timeslot $[\hat{t}, \hat{t}+1]$ is a timeslot where some tiny job reaches its starting deadline for the first time
    \end{enumerate}

It is straightforward to proceed left-to-right and take care of the events (of types 1-4)
where a job arrives or is fulfilled in time $O(n \log n)$,
whether we are in a timeslot covered by a batch, or not.
This processing includes splitting segments and running LeftmostBatch($t$) if necessary, and finding the next event.

If the next event has a batch starting at time $\hat{t}$ (event of type 5 above), then the algorithm schedules all the huge jobs that make the batch
 in a procedure we call ScheduleBatch($\hat{t}$).
To find, based on the current $H(t)$, where the next batch will start, we call a procedure we call LeftmostBatch($t$).
With effort, we will make below 
ScheduleBatch($\hat{t}$) and LeftmostBatch($t$) run in time $O(n \log n)$.
As we show below, every  ScheduleBatch($\hat{t}$) fulfills at least one huge job, and so there will be at most $O(n)$ runs of 
ScheduleBatch($\hat{t}$).

Now we discuss the event of type 6.
If a critical tiny job reaches its starting deadline at a critical timeslot, and there is no available huge job to cover that timeslot,
we create a segment of tiny-only slots, to last until the next event,
which can be computed in time $O(n \log n)$. Also within this time bound, we can find the huge job with earliest deadline that can cover this timeslot,
if such a huge job exists. This creates a segment of covered timeslots that lasts until the next event.

\subsection{LeftmostBatch($t$)}

This procedure is a compressed (segment-by-segment, and not timeslot-by-timeslot)
version of the algorithms of \cite{Glover1967MaximumMI,GALLO198431} run in reverse order. 
We would need to prove some of its properties,
as we are not interested in a maximum convex bipartite matching, but in the leftmost possible batch.
In a high level description,
LeftmostBatch starts with the rightmost deadline of a huge jobs and proceeds leftwards,
repeatedly finding a maximal interval that can
be covered by huge jobs whose deadline is no larger than the rightmost endpoint of the interval.
To find such an interval, we apply a "reverse-EDF" procedure that is a (compressed) reverse
version of the procedure used in the proof of Claim \ref{claim:EDF_works}. 

Precisely, this reverse-EDF procedure works as follows. let $H_i$ be a set of huge jobs.
We let $B_i$ to be a variable in the procedure, and it starts as being the jobs of  $H_i$ sorted by their deadline.
Let $j_i$ be a job of $H_i$ of latest deadline $R_i :=  d(j_i)$.
We have a pending set of jobs $A$ (a variable in the procedure), which initially consists of $j_i$. Remove $j_i$ from $B_i$
(it is at the end of $B_i$).
For each job in $A$, we keep the remaining processing time. 
And we use a max-heap to find from $A$ the job with {\bf latest} arrival time.
Starting from $R_i$, we proceed leftwards, until either the topmost job of the heap is fulfilled, or it reaches or exceeds its arrival time,
or we reach the deadline of the last job in $B_i$ (largest deadline in $B_i$).
If the topmost job of $A$ is fulfilled or it reaches or exceeds its arrival time, we remove it from $A$.
If we reach the deadline of job $j'$ in $B_i$, we update the remaining processing time of the topmost job in $A$, 
add $j'$ to $A$ and remove $j'$ from $B_i$. This process continues until $A$ becomes empty or we reach the time $t$.
Let $B_{i+1}$ be sorted set $B_i$ at the end of the procedure. Let $H_{i+1} = B_{i+1}$. 

If $A$ becomes empty, we have obtained an interval $Z_i$, with right endpoint $R_i$ and left endpoint  $L_i$
such that all the timeslots in $Z_i$ are covered by jobs from $H_i \setminus H_{i+1}$. 
As an aside,
similar to the proof of Claim \ref{claim:EDF_works}, one can prove that $Z_i$ is the largest interval that can be covered
by jobs from $H_i$ among the intervals ending at $R_i$ (recall that $R_i = \max_{j \in H_i} d(j)$).
Moreover, all the jobs in $B_{i+1}$ have deadline at most $L_i - 1$.
The time spent on finding $Z_i$ with its associated jobs is $O( |H_i \setminus H_{i+1}| \cdot \log n )$,
since in each iteration, either one job from $H_i \setminus H_{i+1}$ enters $A$, 
or is removed from $A$,  and each iteration involves a constant number of heap operations.

LeftmostBatch($t$) starts with $H_1 = H(t)$ and applies the reverse-EDF procedure to find $Z_1$, followed by $Z_2, \ldots, Z_i$ as described next.
If we reach $t$ while constructing $Z_i$, we have found a batch as required by the algorithms $G$ or $\hG$. The procedure
LeftmostBatch($t$)  returns $t$.

If $H_{i+1}$ is empty, then  LeftmostBatch($t$) returns $L_i$.
Otherwise, we resume applying the reverse-EDF procedure to $H_{i+1}$.
Based on the time spent on finding each $Z_i$, the overall time of 
LeftmostBatch($t$) is $O(n \log n)$  (slightly generalizing \cite{GALLO198431}).
It is an invariant of LeftmostBatch($t$) that $H_i$ contains exactly all the jobs of $H(t)$ with deadline at most $d(j_i)$.
The correctness of the procedure follows from:

\begin{claim}
    \label{c_leftmost}
    Assume that LeftmostBatch($t$) ends with $H_{i+1} = \emptyset$. Then there is no batch that can be scheduled at a time $t' < L_i$.
\end{claim}
\begin{proof}
    Recall that a batch starting at time $t'$ implies the
    existence of a set $K$ of huge jobs that can cover all timeslots of
    $[t',  d(j^*)]$, where $j^* = \argmax_{j \in K} d(j)$. Suppose (for a contradiction) that such a $K$ exists for $t' < L_i$.
    Let us look at $H_l$ and the interval $Z_l = [L_l,R_l]$  that LeftmostBatch($t$) produces with $j^* \in H_l \setminus H_{l+1}$. 
    From the invariant above we have that $K \subseteq H_l$. 
    Note that $t' < L_i \leq  L_l < d(j^*) \leq R_l = d(j_l)$, since $j^*$ is in $H_l$. 
    Consider the set $K'$ that includes $K$, with the original remaining processing times (as at the beginning of LeftmostBatch($t$),
    and jobs from $H_l \setminus K$ with remaining processing times where 
    we subtract from a jobs' remaining processing time the time
    it would be scheduled by the reverse EDF procedure in the interval $[d(j^*),d(j_l)]$.
    Now consider the run of reverse EDF on $K'$ - it would be the same as the run of reverse EDF on $H_l$ to the left of $d(j^*)$.
    Now modify the schedule of $K$ into timeslots $[t',d(j^*)]$ as in the proof of Claim  \ref{claim:EDF_works} (reverse order)
    to obtain a schedule that matches the reverse EDF schedule of $K'$ mentioned above.
    The result is that the reverse EDF schedule of $K'$ reaches $t'$, a contradiction.
\end{proof}

\subsection{ScheduleBatch($\hat{t}$)}
    LeftmostBatch($t$) already produces a batch that starts at $t$ if there is one: the last nonempty set $H_i$.
    Indeed, if  $j^* = \argmax_{j \in H_i} d(j)$, then  $Z_i = [t, d(j^*)]$ is an interval that is fully covered 
    by the reverse-EDF procedure above with jobs from $H_i$. We must make some modification to this reverse EDF procedure
    since we insist that all the jobs in $H_i$ are fulfilled.

    The modified reverse EDF procedures makes sure that, when scheduling a job $j$ with remaining processing time $rp(j)$
    in a segment with right end $R$, that $rp(j) \leq  R - \max(t,a(j))$.
    Precisely, 
    while running reverse EDF procedures from right to left segment by segment,
    when scheduling a job $j$ with remaining processing time $rp(j)$ in a segment with right end $R$,
    the scheduling of $j$ is kept unmodified if $rp(j) \leq  R - \max(t,a(j))$. 
    
    If this condition does not hold ($rp(j) > R - \max(t,a(j))$),  the modification is:
    job $j$ is scheduled in the interval $[\max(t, a(j)), \max(t, a(j)) + rp(j)]$.
    Since this condition did hold previously (to the right), 
    one can easily check that indeed  this interval
    is  included in $Z_i$ and does not overlap with timeslots where job $j$ is already scheduled. When we do so, we fulfill job $j$.
    Also, the modified reverse EDF procedure will not cover any less than the original reverse EDF procedure.
    This extra step creates at most two extra segments.
    Also, when running the modified EDF procedure, if we reach $t$ and some job $j$ is not fulfilled, we schedule (as above) job $j$
    in the interval $[\max(t, a(j)), \max(t, a(j)) + rp(j)]$. Again,
    since this condition did hold previously (to the right), 
    one can easily check that indeed  this interval
    is  included in $Z_i$ and does not overlap with timeslots where job $j$ is already scheduled. Again, when we do so, we fulfill job $j$.
    And again, this extra step creates at most two extra segments.

    The steps of the original EDF procedure also each introduce at most two new segments. Moreover, whenever a step is executed
    and a job $j$ is schedule in some interval $I$, then this happens at $d(j)$, or $j$ is fulfilled, so in total
    there are $O(n)$ such steps (over all the runs of ScheduleBatch) and so we do have that the number of segments is $O(n)$.

%% file: a_tighter_analysis.tex
\section{The improved guarantee}
\label{s_chatgpt}

Other than this paragraph,
this whole section is written by ChatGPT Pro Lite in October 2026
and carefully checked by the authors, who only made minor edits
regarding notation and organization.

A finer size-class weighting argument improves the guarantee for the \emph{unchanged} offline algorithm from Section \ref{s_offline&online}.
\[
\boxed{\frac{5671}{2340}=2.4235042735\ldots<\frac{22}{9}.}
\]

\begin{theorem}\label{thm:main}
Take the feasible temporal assignment produced by the manuscript's offline algorithm $G$: it has $h$ units of huge-job processing and exactly $s^*$ nonempty timeslots containing no huge job. Run FFD on the original job heights in each timeslot. The resulting total busy time $a$ satisfies
\begin{equation}\label{eq:main}
a\le\frac{5671}{2340}\OPT.
\end{equation}
\end{theorem}

The temporal-scheduling premise is exactly the property used in  Section \ref{s_offline&online}, following Theorem \ref{thm: hG}. The argument here does not replace or reprove the temporal-scheduling theorem or the event-driven implementation. It strengthens the packing analysis, with \emph{no change to $G$, its choices of batches, or FFD}.

\subsection{The weight function}

Set
\begin{equation}\label{eq:alpha-beta}
K=15,\qquad \alpha=\frac{16}{15},\qquad
\beta_i=\frac{15-i}{15i(i+1)}
=\frac1i-\frac{\alpha}{i+1}
\quad(2\le i\le15).
\end{equation}
In particular, $\beta_{15}=0$, and the sequence $\beta_2,\ldots,\beta_{15}$ is nonnegative and decreasing. For $2\le i\le14$, call an item of size in $(1/(i+1),1/i]$ a class-$i$ item; class~$15$ consists of the items of size at most $1/15$.

Define
\begin{equation}\label{eq:weight}
w(x)=
\begin{cases}
\alpha x, & 1/2<x\le1,\\[1mm]
\alpha x+\beta_i, & 1/(i+1)<x\le1/i,\quad 2\le i\le14,\\[1mm]
\alpha x, & 0<x\le1/15.
\end{cases}
\end{equation}
For a finite item list $L$, write $\wt(L)=\sum_{x\in L}w(x)$. These are analysis weights only; the algorithm still packs using the original heights.

Two elementary properties will be useful. A class-$i$ item, for $2\le i\le14$, satisfies
\begin{equation}\label{eq:class-unit-weight}
w(x)>\frac{\alpha}{i+1}+\beta_i=\frac1i.
\end{equation}
Also, for every $m\in\{2,\ldots,15\}$,
\begin{equation}\label{eq:density}
0<x\le\frac1m\quad\Longrightarrow\quad
w(x)\le\frac{m+1}{m}x.
\end{equation}
Indeed, for an item of class $i<15$,
\[
\frac{w(x)}x\le\alpha+(i+1)\beta_i
=\frac{i+1}{i},
\]
and the right-hand side decreases with $i$. For a class-$15$ item, $w(x)/x=16/15$.

\subsection{A stronger FFD lemma}

\begin{lemma}\label{lem:ffd}
Let $L$ be a nonempty bin-packing instance and let $H$ be its number of items strictly larger than $1/2$. For the weight function~\eqref{eq:weight},
\begin{equation}\label{eq:ffd-bound}
\FFD(L)\le
\begin{cases}
1+\wt(L),& H=0,\\
H+\wt(L),& H>0.
\end{cases}
\end{equation}
\end{lemma}

\begin{proof}
Write $M=\FFD(L)$. If $H>0$ and $M=H$, the result is immediate. Otherwise, let $\ell$ be the item that opens bin $M$, and delete every item processed after $\ell$. FFD on the remaining prefix opens the same $M$ bins, its last bin contains exactly $\ell$, and the weight has not increased. It therefore suffices to analyze this prefix. When $H>0$ and $M>H$, no huge item was deleted, since all huge items precede $\ell$ in the sorted order.

From now on, denote the size of the last item by $l$. It is non-huge. Every preceding bin has load strictly larger than $1-l$. Let
\[
r=\min\{15,\lfloor1/l\rfloor\},
\]
so $2\le r\le15$ and $l\le1/r$. All non-huge items have class indices at most $r$.

\paragraph{Regular and exceptional bins.}
The first $H$ bins are the bins opened by huge items. Classify every other bin by the class of its first item. If $i<15$, all bins of type $i$, except possibly the last bin of that type, receive at least $i$ class-$i$ items. To see this, while class-$i$ items are processed, a type-$i$ bin with fewer than $i$ items has load at most $(i-1)/i$ and can accommodate the next class-$i$ item. FFD cannot open another type-$i$ bin while such a bin exists. A bin containing $i$ class-$i$ items has weight at least~$1$, by~\eqref{eq:class-unit-weight}.

If $r<15$, every type-$r$ bin other than the last bin is therefore of weight at least~$1$. If $r=15$, the same conclusion holds for the type-$15$ bins other than the last bin. When $l\le1/16$, each such bin has load greater than $1-l\ge15/16$. When $l>1/16$, each has at least $15$ items: with fewer items its load would be at most $14/15$, so it could accommodate the last item. All its items have size at least $l$, and it again has load greater than $15/16$. Its weight is consequently at least~$1$.

It follows that, apart from the huge-item bins and the last bin, all bins have certified weight at least~$1$, except possibly one bin for each type
\[
i\in\{2,\ldots,r-1\}.
\]
Call those possibly exceptional bins $E_i$, let $S$ be their set of type indices, and put $q=|S|$. Bins already certified to have weight at least~$1$ will be called regular.

Each $E_i$ contains at least $i$ items in total. Otherwise its load would be at most $(i-1)/i$, and the last item, which has size at most $1/i$, would have fitted in it. At least one of its items is of class $i$, and the others have class indices at most $r$. The total bonus, meaning weight minus $\alpha$ times volume, in $E_i$ is therefore at least
\[
\beta_i+(i-1)\beta_r.
\]
The last bin has bonus $\beta_r$. Thus the total bonus in the exceptional bins and the last bin is at least
\begin{equation}\label{eq:bonus-pool}
B(S,r):=\beta_r+\sum_{i\in S}\bigl(\beta_i+(i-1)\beta_r\bigr).
\end{equation}

\paragraph{A finite arithmetic inequality.}
For all $2\le r\le15$ and all $S\subseteq\{2,\ldots,r-1\}$, we have
\begin{equation}\label{eq:pool-certificate}
\alpha\left(q+1-\frac qr\right)+B(S,r)\ge q+1.
\end{equation}
Here is a complete certificate. Set
\[
\tau_r=1-\alpha\left(1-\frac1r\right)
=\frac{16-r}{15r}>0.
\]
The left-hand side of~\eqref{eq:pool-certificate}, minus its right-hand side, equals
\[
\frac1{15}+\beta_r+
\sum_{i\in S}\left(\beta_i+(i-1)\beta_r-\tau_r\right).
\]
The minimum over all subsets $S$ is exactly
\begin{equation}\label{eq:D}
D_r:=\frac1{15}+\beta_r+
\sum_{i=2}^{r-1}\min\left\{0,\beta_i+(i-1)\beta_r-\tau_r\right\}.
\end{equation}
The following table lists every value; all are positive. In the middle column, an interval of integers includes both endpoints.

\begin{center}
\renewcommand{\arraystretch}{1.14}
\begin{tabular}{@{}rll@{}}
\toprule
$r$ & Indices of negative terms in~\eqref{eq:D} & $D_r$\\
\midrule
2  & $\varnothing$ & $19/90$\\
3  & $\{2\}$ & $1/18$\\
4  & $\{2,3\}$ & $11/450$\\
5  & $\{3,4\}$ & $1/100$\\
6  & $\{3,4,5\}$ & $11/6300$\\
7  & $\{4,5,6\}$ & $41/6300$\\
8  & $\{4,\ldots,7\}$ & $221/37800$\\
9  & $\{4,\ldots,8\}$ & $17/1400$\\
10 & $\{5,\ldots,9\}$ & $7303/415800$\\
11 & $\{6,\ldots,10\}$ & $10663/415800$\\
12 & $\{6,\ldots,11\}$ & $173479/5405400$\\
13 & $\{7,\ldots,12\}$ & $212089/5405400$\\
14 & $\{8,\ldots,13\}$ & $248257/5405400$\\
15 & $\{10,\ldots,14\}$ & $46691/900900$\\
\bottomrule
\end{tabular}
\end{center}
This is a finite rational-arithmetic verification of an explicitly stated inequality, not an assumption based on sampling FFD instances.

\paragraph{The case $H>0$.}
Use one huge-item bin, the $q$ exceptional bins, and the last bin as a pool. The first $q+1$ bins in this pool have load greater than $1-l$, and the last bin has load $l$. The pool's total volume is greater than
\[
(q+1)(1-l)+l=q+1-ql\ge q+1-\frac qr.
\]
Its bonus is at least $B(S,r)$. By~\eqref{eq:pool-certificate}, its weight is at least $q+1$. The $M-H-q-1$ regular non-huge bins contribute at least that many additional units of weight. Ignoring the nonnegative weights of the other huge-item bins gives
\[
\wt(L)\ge(q+1)+(M-H-q-1)=M-H.
\]
This proves the second inequality of~\eqref{eq:ffd-bound}.

\paragraph{The case $H=0$.}
If $q=0$, the $M-1$ bins other than the last bin already each have weight at least~$1$. Suppose $q\ge1$. Pool the $q$ exceptional bins and the last bin. Its volume is greater than
\[
q(1-l)+l=q-(q-1)l\ge q-\frac{q-1}{r}.
\]
Its weight, minus $q$, is at least
\begin{align*}
&\alpha\left(q-\frac{q-1}{r}\right)+B(S,r)-q\\
&\quad=\left[\alpha\left(q+1-\frac qr\right)+B(S,r)-(q+1)\right]
+1-\alpha\left(1-\frac1r\right)\\
&\quad\ge D_r+\tau_r\ge0.
\end{align*}
The remaining $M-q-1$ bins are regular, yielding $\wt(L)\ge M-1$. This proves the first inequality.
\end{proof}

If zero-height jobs are allowed, set $w(0)=0$. Zero-size items cannot create an additional FFD bin in a list that contains a positive-size item. An all-zero nonempty list uses one bin and satisfies the $H=0$ inequality separately. Thus they do not affect the lemma or its scheduling application.

\subsection{Weight of one feasible machine-timeslot}

\begin{lemma}\label{lem:opt-bin}
For every finite multiset $B$ of item sizes with $\sum_{x\in B}x\le1$,
\begin{equation}\label{eq:one-bin}
\wt(B)\le\frac{3331}{2340}.
\end{equation}
\end{lemma}

\begin{proof}
Write $V=\sum_{x\in B}x\le1$. We distinguish configurations using the manuscript's definition of big-but-not-huge items: sizes in $(1/3,1/2]$. In the proof, a medium item has size in $(1/4,1/3]$.

\paragraph{A huge item is present.}
Let its size be $x>1/2$. It is unique. All other items have size at most $1/2$, so~\eqref{eq:density} gives
\[
\wt(B)\le\alpha x+\frac32(1-x)
\le\frac{\alpha}{2}+\frac34
=\frac{77}{60}<\frac{3331}{2340}.
\]

\paragraph{No huge item and no big item.}
All items have size at most $1/3$. Therefore $\wt(B)\le4/3$.

\paragraph{Exactly one big item.}
Let its size be $x>1/3$. The other items have size at most $1/3$. Since $\alpha<4/3$,
\begin{align*}
\wt(B)&\le\alpha x+\beta_2+\frac43(1-x)\\
&\le\frac43+\beta_2+\frac13\left(\alpha-\frac43\right)
=\frac{25}{18}<\frac{3331}{2340}.
\end{align*}

\paragraph{Two big items, and no medium item.}
Let $x$ be the sum of the two big sizes, so $x>2/3$. All remaining items have size at most $1/4$. Since $\alpha<5/4$,
\begin{align*}
\wt(B)&\le\alpha x+2\beta_2+\frac54(1-x)\\
&\le\frac54+2\beta_2+\frac23\left(\alpha-\frac54\right)
=\frac{17}{12}<\frac{3331}{2340}.
\end{align*}

\paragraph{Two big items and one medium item.}
There cannot be two medium items, since $2/3+2/4>1$. Let $x$ be the total size of the two big items and the one medium item. Then
\[
x>\frac23+\frac14=\frac{11}{12}.
\]
Consequently the total size of all remaining items is less than $1/12$.

First suppose that all remaining items have size at most $1/13$. By~\eqref{eq:density} with $m=13$, and since $\alpha<14/13$,
\begin{align*}
\wt(B)&\le\alpha x+2\beta_2+\beta_3+\frac{14}{13}(1-x)\\
&\le\alpha\frac{11}{12}+2\beta_2+\beta_3+\frac{14}{13}\frac1{12}\\
&=\frac43+\frac7{78}
=\frac{37}{26}<\frac{3331}{2340}.
\end{align*}

Otherwise there is exactly one remaining item larger than $1/13$. It has size less than $1/12$, so it is a class-$12$ item. After that item, the remaining total size is less than
\[
\frac1{12}-\frac1{13}=\frac1{156}<\frac1{15}.
\]
All those remaining items have zero bonus. Hence
\begin{align*}
\wt(B)
&=\alpha V+2\beta_2+\beta_3+\beta_{12}\\
&\le\alpha+2\beta_2+\beta_3+\beta_{12}\\
&=\frac{16}{15}+\frac{13}{45}+\frac1{15}+\frac1{780}
=\frac{3331}{2340}.
\end{align*}
The cases are exhaustive because three big items cannot fit in one bin.
\end{proof}

The last configuration also explains the choice of size thresholds. Two items just above $1/3$ leave less than $1/3$; an item just above $1/4$ leaves less than $1/12$; an item just above $1/13$ leaves less than $1/156$. This is the residual capacity that the coarse $b,v$ analysis does not track.

For the particular weight function~\eqref{eq:weight}, the constant $3331/2340$ in Lemma~\ref{lem:opt-bin} is attained. Take two sizes $1/3+\varepsilon$, one size $1/4+\varepsilon$, one size $1/13+\varepsilon$, and one size $1/156-4\varepsilon$, with sufficiently small positive $\varepsilon$. Their sizes sum to~$1$, and their bonuses are exactly $2\beta_2+\beta_3+\beta_{12}$. This does \emph{not} assert tightness of the scheduling approximation ratio.

\subsection{Summing the inequalities over time}

\begin{proof}[Proof of Theorem~\ref{thm:main}]
For a job $j$ of height $h(j)$ and processing time $p(j)$, its contribution to the total weight of processing units is
\[
p(j)w(h(j)).
\]
Define the schedule-independent quantity
\begin{equation}\label{eq:global-weight}
W:=\sum_j p(j)w(h(j)).
\end{equation}

Apply Lemma~\ref{lem:ffd} separately in each nonempty timeslot of the temporal schedule produced by $G$. A timeslot containing huge jobs contributes its number of huge job units plus the weight of its assigned units. A nonempty timeslot without a huge job contributes $1$ plus its assigned weight. There are exactly $s^*$ timeslots of the latter kind, and the total number of huge job units is $h$. Therefore
\begin{equation}\label{eq:algorithm-sum}
a\le h+s^*+W.
\end{equation}

Now split an optimal original schedule into unit machine-timeslots. There are $\OPT$ such bins, and each is a feasible capacity-one bin of job heights. Lemma~\ref{lem:opt-bin} gives
\begin{equation}\label{eq:opt-sum}
W\le\frac{3331}{2340}\OPT.
\end{equation}
The Huge-Tiny instance is a relaxation of the original instance, so the manuscript's temporal optimum satisfies $h+s^*\le\OPT$. Combining this with~\eqref{eq:algorithm-sum} and~\eqref{eq:opt-sum} proves
\[
a\le\left(1+\frac{3331}{2340}\right)\OPT
=\boxed{\frac{5671}{2340}\OPT}.
\]
\end{proof}

Numerically, the improvement over the manuscript's bound is
\[
\frac{22}{9}-\frac{5671}{2340}=\frac{49}{2340}
=0.0209401709\ldots.
\]
It is an absolute approximation guarantee: there is no additive error term and no assumption that $\OPT$ is large. The size classes and weights are used only in the proof, so no additional computation is required in the scheduling algorithm.